\documentclass[12pt]{article}

\usepackage{amssymb}
\usepackage{amsmath}
\usepackage{amsthm}
\usepackage{mathrsfs}
\usepackage{MnSymbol}
\usepackage{cite}
\usepackage{graphicx}
\usepackage{appendix}
\usepackage{enumerate}
\usepackage{calligra}
\DeclareMathAlphabet{\mathcalligra}{T1}{calligra}{m}{n}
\DeclareFontShape{T1}{calligra}{m}{n}{<->s*[2.2]callig15}{}

\usepackage[cal=boondox]{mathalfa}

\newcommand{\R}{\ensuremath{\mathbb{R}}}
\newcommand{\RP}{\ensuremath{\mathbb{R}_+}}
\newcommand{\scrp}{\ensuremath{\mathcal{p}}}
\newcommand{\scrq}{\ensuremath{\mathcal{q}}}
\newcommand{\scrm}{\ensuremath{\mathcal{m}}}

\newcommand{\scru}{\ensuremath{\mathcal{u}}}
\newcommand{\scrv}{\ensuremath{\mathcal{v}}}
\newcommand{\scrw}{\ensuremath{\mathcal{w}}}

\newcommand{\scrP}{\ensuremath{\mathscr{P}}}
\newcommand{\scrPhat}{\ensuremath{\hat{\mathscr{P}}}}

\newcommand{\scrC}{\ensuremath{\mathscr{C}}}
\newcommand{\scrR}{\ensuremath{\mathscr{R}}}

\newcommand{\scrK}{\ensuremath{\mathscr{K}}}
\newcommand{\scrV}{\ensuremath{\mathscr{V}}}
\newcommand{\scrM}{\ensuremath{\mathscr{M}}}

\newcommand{\scrT}{\ensuremath{\mathscr{T}}}
\newcommand{\scrH}{\ensuremath{\mathscr{H}}}
\newcommand{\scrQ}{\ensuremath{\mathscr{Q}}}
\newcommand{\supp}{\ensuremath{\mathrm{supp} \,}}

\newcommand{\Cone}{\ensuremath{\mathrm{Cone} \,}}

\newcommand{\coneP}{\ensuremath{\textnormal{Cone\,(\scrP)}}}
\newcommand{\theory}{\ensuremath{(\Sigma,\scrP)}}
\newcommand{\deltam}{\ensuremath{\Delta\scrm}}
\newcommand{\process}{\ensuremath{(\Delta\scrm,\scrq)}}
\newcommand{\MSigma}{\ensuremath{\scrM(\Sigma)}}
\newcommand{\MSigmaZ}{\ensuremath{\scrM^{\circ}(\Sigma)}}
\newcommand{\MSigmaPl}{\ensuremath{\scrM_{+}(\Sigma)}}
\newcommand{\MSigmaPlOne}{\ensuremath{\scrM_{+}^1(\Sigma)}}
\newcommand{\VSigma}{\ensuremath{\scrV(\Sigma)}}
\newcommand{\etao}{\ensuremath{\eta^{\,0}}}
\newcommand{\To}{\ensuremath{T^{\,0}}}

\theoremstyle{plain}
\newtheorem{theorem}{Theorem}[section]

\newtheorem{lemma}[theorem]{Lemma}
\newtheorem{proposition}[theorem]{Proposition}
\newtheorem{corollary}[theorem]{Corollary}
\theoremstyle{definition}
\newtheorem{definition}[theorem]{Definition}
\newtheorem{example}[theorem]{Example}

\theoremstyle{remark}
\newtheorem{rem}[theorem]{Remark}

\newcounter{property}

\begin{document}

\author{Martin Feinberg\thanks{The William G. Lowrie Department of Chemical \& Biomolecular Engineering and Department of Mathematics, The Ohio State University, 151 W. Woodruff Avenue, Columbus, OH 43210 USA.  E-mail: feinberg.14@osu.edu.}
\and Richard B. Lavine\thanks{Department of Mathematics, University of Rochester, Rochester, NY 14627 USA.  E-mail: rdlavine@frontiernet.net
}}

\title{Entropy and Thermodynamic Temperature in Nonequilibrium Classical Thermodynamics as Immediate Consequences of the Hahn-Banach Theorem: II. Properties}

\date{\emph{\today}}
\maketitle
\begin{abstract} 
In a companion article it was shown in a certain precise sense that, for any thermodynamical theory that respects the Kelvin-Planck Second Law, the Hahn-Banach Theorem immediately ensures the existence of a pair of continuous functions of the local material state --- a specific entropy (entropy per mass) and a thermodynamic temperature --- that together satisfy the Clausius-Duhem inequality for every process. There was no requirement that the local states considered be states of equilibrium.  This article addresses questions about properties of the entropy and thermodynamic temperature functions so obtained: To what extent do such temperature functions provide a faithful reflection of ``hotness"? In precisely which Kelvin-Planck theories is such a temperature function essentially unique, and, among those theories, for which is the entropy function also essentially unique? What is a thermometer for a Kelvin-Planck theory, and, for the theory, what properties does the existence of a thermometer confer? In all of these questions, the Hahn-Banach Theorem again plays a crucial role.
\end{abstract}

\newpage
\tableofcontents
\newpage
\numberwithin{equation}{section}

\section{Introduction}

In a companion article \cite{feinberg-lavineEntropy1} we showed in a certain precise sense that, for any thermodynamical theory in which all processes comply with the Kelvin-Planck Second Law of Thermodynamics, there must exist a pair of continuous functions of the local material state --- a specific entropy (entropy per mass) and a thermodynamic temperature --- that together satisfy the Clausius-Duhem inequality for every process. That this is so is an \emph{immediate} consequence of the Hahn-Banach Theorem. For \emph{existence} of these  functions \emph{there is no reliance at all on the presence of special processes, such as very slow reversible ones (e.g., Carnot cycles)}.

In this article the situation will be different. Once again, the Hahn-Banach Theorem will be the central tool,\footnote{The Hahn-Banach Theorem will usually exert itself in the guise of Lemma \ref{lemma:FirstBigLem} or \ref{lem:CDOpposite}.} but this time in examining, for any given  Kelvin-Planck theory,  properties of the Clausius-Duhem entropy-temperature pairs that the theory admits. In almost every theorem contained here we show that, in sharp contrast to the existence theorem, the presence of special processes within the theory is not only sufficient to ensure that the temperature and entropy functions have specific properties \emph{but also necessary}. 

The difference resides in the fact that there is  an inverse relationship between the supply of processes the theory contains and the supply of entropy-temperature pairs that  satisfy the Clausius-Duhem inequality for all such processes. The larger the supply of processes, the smaller the set of Clausius-Duhem entropy-temperature pairs, and vice-versa.  Thus, if the set of entropy-temperature pairs for a given Kelvin-Planck theory is to have a particular property (e.g., essential uniqueness), then the set of processes extant in the theory must be sufficiently large as to ensure that the theory's set of entropy-temperature pairs is suitably narrow. Stipulating the required breadth of the process supply will often amount to specifying that it \emph{must} contain an abundance of processes of a very particular kind (e.g., Carnot cycles).

The brilliant 19$^{\text{th}}$ century pioneers invoked an abundance of  idealized reversible Carnot cycles (traversing only equilibrium states) to deduce, almost simultaneously, \emph{both} the existence \emph{and} the uniqueness of entropy and thermodynamic temperature functions (defined on those states). It is understandable, then, that  the classical arguments might lead to a conflation of, on one hand, existence and uniqueness and, on the other hand, necessary and sufficient conditions for each of these. Although an abundance of Carnot cycles visiting equilibrium states might be \emph{sufficient}, as argued by the pioneers, for the \emph{existence} of entropy and thermodynamic temperature functions, the presence of those cycles is \emph{not necessary}, as we showed in  \cite{feinberg-lavineEntropy1} (under substantially weaker conditions than in \cite{feinberg1986foundations}), nor is there a necessity to restrict the domains of those functions to equilibrium states. Here, among other things, we will develop far more fully ideas in \cite{feinberg1983thermodynamics} and \cite{feinberg1986foundations} to the effect that \emph{particular properties} of the temperature and entropy functions (e.g., uniqueness) actually \emph{require}, remarkably, the presence of something like those  specific processes the pioneers imagined.
\subsection{How This Article is Structured}
	After providing a synopsis of \cite{feinberg-lavineEntropy1} in Section \ref{sec:Synopsis},  we examine in Section \ref{sec:Hotness} the relationship between temperature and hotness. In particular, for a Kelvin-Planck theory we posit definitions, \emph{stated solely in terms of the processes extant within the theory}, of what it means for one local material state to be of the same hotness as another  state and for one to be hotter than another.  We then show that these relations are \emph{precisely} reflected in the set of Clausius-Duhem temperature scales the theory admits. Not only does each temperature scale reflect the hotness relation faithfully, but also if one state has a higher temperature than another on each Clausius-Duhem temperature scale, then the set of processes in the theory must be sufficiently structured as to establish that the first state is indeed hotter than the second. In this implication the Hahn-Banach Theorem plays a crucial role. 
	
	Sections \ref{sec:TempUniqueness} and \ref{sec:EntropyUniqueness} are largely devoted to uniqueness questions. Uniqueness of a Clausius-Duhem temperature scale for a Kelvin-Planck theory is taken up in Section \ref{sec:TempUniqueness}, where we show, among other things, that for such a scale to be essentially unique, not only is it sufficient that the set of processes extant in the theory be abundantly rich in Carnot cycles \emph{but also necessary}. For a Kelvin-Planck theory having an essentially unique Clausius-Duhem temperature scale, we ask in Section \ref{sec:EntropyUniqueness} about circumstances under which its companion Clausius-Duhem specific-entropy function is also essentially unique. Among other things, we show that for essential uniqueness of entropy on the entire state space domain, it is not only sufficient that any two states be connected by a reversible process \emph{but also necessary}. Here again, the Hahn-Banach Theorem is crucial.

	In Section 6 we take cognizance of the fact that two very different bodies---one perhaps a metal rod and the other a liquid solution exhibiting chemical reactions and diffusion, embraced within two very different Kelvin-Planck theories---can exchange heat with each other. With this in mind, we study in Section 6  properties of  a ``conjoined" Kelvin-Planck theory that subsumes smaller distinct ones. Our special focus is on  conjunctions of separate Kelvin-Planck theories of two  different materials,  wherein the first material can serve as a (suitably defined) \emph{thermometer} for the second. In that case, we study how the larger conjoined theory can impart to the second Kelvin-Planck theory  additional ``hotter than" relations and uniqueness properties that were not intrinsic to it.

	Section 7 contains concluding remarks. With an eye toward clarifying and softening distinctions that are sometimes drawn between ``equilibrium" and "non-equilibrium"  thermodynamics, we  review, among other things, what the theorems in this article tell us about the (sometimes conflated) necessary and sufficient conditions for the very separate questions of existence and uniqueness of Clausius-Duhem entropy-temperature pairs.

\section{Synopsis of Part I} \label{sec:Synopsis}

 A \emph{thermodynamical theory} in \cite{feinberg-lavineEntropy1} is an abstraction of  just those features of a  material (or collection of materials) that bear upon statements of the Second Law and, ultimately,  upon statements of the Clausius-Duhem inequality. A particular theory is indicated by a pair of sets, \theory, with $\Sigma$ denoting the theory's \emph{state space} and \scrP\ denoting the set of \emph{processes} experienced by those material bodies the theory is deemed to describe.

Elements of $\Sigma$ are understood to represent  the possible \emph{local} states that might be exhibited in bodies as they experience physical processes captured in \scrP.  The state space will often take the form of a subset of $\mathbb{R}^n$.  In a theory of particular gas, for example, the states might consist of pairs  $[\,p,v] \in \mathbb{R}^2$, where $p$ is the local pressure and $v$ is the local specific volume (the reciprocal of the density). In a theory of a reacting mixture of $n$ chemical species, the states might consist of vectors $[\,c_1, c_2,\dots, c_n, \theta] \in \mathbb{R}^{n+1}$, where $c_i$ is the local molar concentration of the $i^{th}$ species and $\theta$ is the local temperature on some empirical temperature scale. For reasons given in \cite{feinberg-lavineEntropy1} we assume that $\Sigma$ is a compact Hausdorff space.

	A process in \scrP\, is specified by a pair of objects, the \emph{change of condition} for the process, usually denoted \deltam, and the \emph{heating measure} for the process, usually denoted \scrq. For the purposes of interpretation, imagine a body experiencing a physical process that initiates at time $t_i$ and terminates at a final time $t_f$. We begin by indicating what we mean by the body's \emph{condition} at a given instant and then the body's change of condition over the course of the process.
	
	At each instant $t$ during the course of the process, the $condition$ of the body, $\scrm_t$, is a positive Borel measure on $\Sigma$ with the following meaning: for each Borel set $\Lambda \subset \Sigma$, $\scrm_t(\Lambda)$ is the mass of that part of the body consisting of material in local states residing in $\Lambda$. Note that $\scrm_t(\Sigma)$ is the mass of the entire body at time $t$. For the process, the change of condition is given by the signed Borel measure $\deltam := \scrm_{t_f} - \scrm_{t_i}$. From mass conservation it follows that $\deltam(\Sigma) = 0$.
	
	During the course of the process the body experiencing it will exchange heat with the body's exterior. The heating measure \scrq\ for the process is again a signed Borel measure on $\Sigma$ with the following interpretation: For each Borel set $\Lambda \subset \Sigma$, $\scrq(\Lambda)$ is the net amount of heat absorbed from the body's exterior, over the course of the entire process,  by material which, \emph{at the time of absorption}, is in states within $\Lambda$. Note that $\scrq(\Sigma)$ is the net amount of heat absorbed by the body between the inception of the process and its end.
	
	A member of the process set \scrP\ for the theory \theory\  is then identified with an element of the form \scrp := \process. We denote by \MSigma\  the vector space of regular signed measures on $\Sigma$, taken with its weak-star topology and by \MSigmaZ\  the linear subspace of \MSigma\  consisting of all its members which, like \deltam,  take the value $0$ on $\Sigma$; the topology of  \MSigmaZ\  is the one it inherits as a subset of \MSigma. 
	
	We hereafter regard \scrP\ to be a subset of the vector space $\VSigma := \MSigmaZ \oplus \MSigma$, taken with the product topology. By \coneP\  we mean the set of all non-negative multiples of members of \scrP. In the appendix of  \cite{feinberg-lavineEntropy1} we gave reasons to presume that, in theories that describe natural   physical processes, the closure of \coneP\  should be convex. The discussion so far is summarized in the following definition.

\begin{samepage}
\begin{definition} A \textbf{thermodynamical theory} consists of a (compact) Hausdorff set $\Sigma$, called the \emph{state space} of the theory, and  a set $\scrP \subset \VSigma$ such that 
\begin{equation}
\hat{\scrP} := \textrm{cl}\,(\Cone(\scrP))
\end{equation}
is convex. Elements of \scrP\  are the \emph{processes} of the theory,
\end{definition}
\end{samepage}

	 The Kelvin-Planck version of the Second Law asserts, in effect, that it is impossible in a cyclic process for the body experiencing the process to merely absorb heat from its exterior; it must also emit heat to the exterior, in a manner that is qualitatively different from the heat absorption.

	We say that a thermodynamical theory \theory\ is a \emph{Kelvin-Planck theory} if, in a certain precise sense, it complies with the Kelvin-Planck requirement. We regard a process $\scrp = \process$\ in \scrP\ to be \emph{cyclic} if $\deltam = 0$ --- that is, if the condition of the body experiencing the process is the same at the process's beginning and end.  By \MSigmaPl\  we mean the the subset of \MSigma\  consisting of measures that are non-negative on every Borel set. By $(0,\MSigmaPl)$ we mean the set of all members of \VSigma\ of the form $(0,\nu)$, with $\nu \in \MSigmaPl$. We take the  Kelvin-Planck stricture to require   that \scrP\ meet  $(0,\MSigmaPl)$ at most in $(0,0)$; that is, if $(0,\scrq) \in \scrP$ is a cyclic process such that the heating measure \scrq\  is positive on some Borel set in $\Sigma$, then \scrq\ should be negative on some other Borel set. In fact, for reasons explained in \cite{feinberg-lavineEntropy1}, we  also require a little more:
	
\begin{definition} A \textbf{Kelvin-Planck theory} is a thermodynamical theory \theory\ such that
\begin{equation}
\label{eq:KPSecLaw}
\hat{\scrP}\ \cap\ (0,\,\MSigmaPl) = (0,0).
\end{equation}
\end{definition}

\medskip
Equation \eqref{eq:KPSecLaw} amounts to a requirement that  that no nonzero element of the forbidden cone (0,\,\MSigmaPl) is  approximated by vectors  of \VSigma\ that point along directions associated with members of \scrP. 

	The Hahn-Banach Theorem almost immediately leads to the existence, for a Kelvin-Planck theory, of continuous specific-entropy and thermodynamic temperature functions of state that together comply with the Clausius-Duhem inequality for all processes the theory contains. The version\footnote{See Theorem 1.7 in \cite{brezis2011functional}, Theorem 21.12 in \cite{choquet1969lectures} or Corollary 14.4 in \cite{kelley1963linear}} of the Hahn-Banach Theorem we employ is given below.
	
\begin{theorem}{\emph{\textbf{(Hahn-Banach)}}}\label{thm:HahnBanach} Let $V$ be a Hausdorff locally convex topological vector space, and let $A$ and $B$ be non-empty disjoint closed convex subsets of $V$, with $B$ compact. There is a  continuous linear function $f: V \to\, \mathbb{R}$ and a number $\gamma \in \mathbb{R}$ such that
\begin{equation}
f(a)\  <\  \gamma,\  \forall\  a \in A\nonumber
\end{equation}
and
\begin{equation}
f(b)\  >\  \gamma,\  \forall\  b \in B\nonumber.
\end{equation}
In particular, if $A$ is a cone, then
\begin{equation}
f(a)\  \leq\  0,\  \forall\  a \in A\nonumber
\end{equation}
and
\begin{equation}
f(b)\  >\  0,\  \forall\  b \in B\nonumber.
\end{equation}
\end{theorem}

\medskip

	Unlike classical arguments for the existence of continuous entropy and thermodynamic functions of state, the following theorem --- the most important in this two-part series --- requires nothing in the way of special processes or the idea of equilibrium states. The proof \cite{feinberg1986foundations,feinberg-lavineEntropy1} follows from the Hahn-Banach Theorem directly. In the theorem statement, $\mathbb{R}_+$ denotes the strictly positive real numbers.

\begin{samepage}
\begin{theorem}[Existence of Entropy and Thermodynamic Temperature]
\label{thm:ExistTempEnt}
 For a thermodynamical theory \theory\ the following are equivalent:
\begin{enumerate}[(i)]
\item \theory\ is a Kelvin-Planck theory.
\item There exist functions $\eta \in \textnormal{C}(\Sigma,\mathbb{R})$ and $T  \in \textnormal{C}(\Sigma,\mathbb{R}_+)$  such that
\begin{equation}\label{eq:MainThmCDIneq}
\int_{\Sigma}\eta\: d(\deltam)\  \geq\  \int_{\Sigma}\frac{d\scrq}{T}, \quad \forall\  (\deltam,\scrq) \in \scrP.
\end{equation}
\end{enumerate}
\end{theorem}
\end{samepage}

\begin{definition}[\emph{Entropy, Thermodynamic Temperature}] 
\label{def:CDpair}
Let \theory\ be a Kelvin-Planck theory. An element $(\eta,T)$ of $\textnormal{C}(\Sigma,\mathbb{R}) \times \textnormal{C}(\Sigma,\mathbb{R}_+)$ that satisfies \eqref{eq:MainThmCDIneq} is a \textbf{Clausius-Duhem pair} for the theory. A function $T \in \textnormal{C}(\Sigma,\mathbb{R}_+)$ is a \textbf{Clausius-Duhem temperature scale} for the theory if there exists $\eta \in \textnormal{C}(\Sigma,\mathbb{R})$ such that $(\eta,T)$ is a Clausius-Duhem pair. In that case, $\eta(\cdot)$ is a \textbf{specific-entropy function} for the theory (corresponding to the Clausius-Duhem temperature scale $T(\cdot))$. The set of all Clausius-Duhem temperature scales for the Kelvin-Planck theory \theory\ is denoted $\scrT_{CD}(\Sigma,\scrP)$ or merely $\scrT_{CD}$ when the theory under consideration is apparent.
\end{definition} 

\begin{rem}\label{rem:CD-clP} It will be useful to record for future use an observation made in the proof \cite{feinberg-lavineEntropy1} of Theorem \ref{thm:ExistTempEnt}: If $(\eta,T)$ is a Clausius-Duhem pair for the Kelvin-Planck theory \theory\ --- that is, if it satisfies the inequality in \eqref{eq:MainThmCDIneq} for all members of \scrP\ --- then it also satisfies that inequality for all members of \scrPhat.
\end{rem}

\begin{rem}[\emph{Reversible members of} \scrPhat]
\label{rem:Reversible}
 Note that if\, $\process$ and $-\process$ are both members of \scrP\  (or, more generally, of \scrPhat) then for each Clausius-Duhem pair $(\eta,T)$ we must actually have the equality
\begin{equation}\label{eq:MainThmCDEQUAL}
\int_{\Sigma}\eta\: d(\deltam)\ =  \int_{\Sigma}\frac{d\scrq}{T}.
\end{equation}
\end{rem}
\medskip

\begin{rem}[\emph{Essential uniqueness}] \label{rem:CDPairNonUnique} It is not difficult to see that if, for a thermodynamical theory, $(\eta(\cdot),T(\cdot))$ is a Clausius-Duhem pair then, for any choice of $\alpha \in \RP$ and $\beta \in \R$,
\begin{equation}
\nonumber
(\frac{1}{\alpha}\eta(\cdot) + \beta,\alpha T(\cdot))
\end{equation}
is also a Clausius-Duhem pair. In particular, if $T(\cdot)$ is a Clausius-Duhem temperature scale for a thermodynamical theory, then so is any positive multiple of $T(\cdot)$.  

	However, there might be still other Clausius-Duhem temperature scales that are not of this kind. For this reason, we say that, for a thermodynamical theory, a Clausius-Duhem temperature scale $T(\cdot)$ is \emph{essentially unique} if every other Clausius-Duhem temperature scale for the theory is a positive multiple of $T(\cdot)$. Similarly, if $\eta(\cdot)$ is a Clausius-Duhem specific-entropy function corresponding to a particular Clausius-Duhem temperature scale $T(\cdot)$, we say that  $\eta(\cdot)$ is \emph{essentially unique} if any other Clausius-Duhem entropy scale corresponding to  $T(\cdot)$ differs from  $\eta(\cdot)$ by at most a constant. 
		
	Among other things, we will take up uniqueness questions in the remainder of this article.
\end{rem}
\medskip

\section[Hotness and Its Reflection in Temperature Scales]{Hotness and Its Reflection in  Thermodynamic Temperature Scales} \label{sec:Hotness}

	Theorem \ref{thm:ExistTempEnt} asserts the equivalence of a version of the Kelvin-Planck Second Law and the existence of functions-of-state pairs, consisting of a specific entropy and a thermodynamic temperature, which together satisfy the Clausius-Duhem inequality for all processes the theory contains. Note, however, that as yet there has been no notion of ``hotness" posited for a particular state (as distinct from its temperature), nor has there been posited a meaning for the idea that one state is ``hotter than" another.  
	
	Yet, \emph{hotness} and \emph{hotter than} are notions generally regarded to be inextricable from thermodynamics itself. (Indeed, Clausius's formulation of the Second Law, unlike the Kelvin-Planck formulation, takes ``hotter than" as a primitive idea: \emph{Heat can never pass from a colder to a warmer body without some other change, connected therewith, occurring at the same time}\cite{clausius1854veranderte,pippard1964elements}.)
	
		Temperature is of course supposed to provide a faithful reflection of hotness (the more fundamental notion), and so it should be with the Clausius-Duhem temperature scales derived from the Kelvin-Planck Second Law. To determine whether this is indeed the case, it will be necessary to first posit for a Kelvin-Planck theory \theory\  means by which two states in $\Sigma$ can be judged to be of the same hotness or by which one can be judged hotter than the other. \emph{We take the view that, to the extent that a Kelvin-Planck theory \theory\  can be deemed self-contained, such judgments should derive from examination of the processes described by \scrP\  (or of the elements in its extension \scrPhat).} This view will broaden somewhat in Section \ref{sec:Conjoined} when we consider conjoined thermodynamical theories and the idea of a thermometer.
\medskip

\begin{rem}[\emph{Clausius-Duhem vs$.$ Clausius temperature scales}] \label{rem:ClausVsCDTempScales} In the next few sections we will examine the relationship between hotness and its reflection in Clausius-Duhem temperature scales. Similar questions were addressed in \cite{feinberg1983thermodynamics}, where our concern was with what we called \emph{Clausius temperature scales}. A Clausius temperature scale is one that satisfies the \emph{Clausius inequality}, which is the Clausius-Duhem inequality \eqref{eq:MainThmCDIneq}  \emph{restricted to cyclic processes}. In that case, the left side of \eqref{eq:MainThmCDIneq} reduces to zero ($\Delta \scrm = 0$).  Neither the entropy nor the change of condition plays a role. 

	For this reason, the set of Clausius temperature scales can, for a Kelvin-Planck theory, be different from its set of Clausius-Duhem temperature scales (Definition \ref{def:CDpair}). The delicate relationship between the two sets is discussed in Appendix \ref{app:ClausiusVsCD}. Although theorems in the coming subsections resemble some in \cite{feinberg1983thermodynamics} about Clausius scales, it should be kept in mind that here they are about sets of thermodynamic temperature scales  different from those in \cite{feinberg1983thermodynamics}.
\end{rem}

\subsection{Hotness as Revealed by Processes}

	In deciding for a Kelvin-Planck theory \theory\ the relative hotnesses of two states in $\Sigma$ we will rely heavily on the \emph{cyclic} processes contained in \scrP\ or, more generally, on elements of the form $(0,\scrq)$ in \scrPhat. This is because, after the cycle,  the condition of the body suffering the process is left unchanged and, as a result,  the relationship between heat and work is especially simple. This is explained in the following remark.

\begin{rem} Although the First Law of Thermodynamics plays no formal role here, one aspect of it will help guide our consideration of hotness. If, in a thermodynamical theory \theory, \process\  is a process, then $\scrq(\Sigma)$ is the \emph{net} amount of heat absorbed by the body experiencing the process over the entire course of the process. If the process is cyclic (i.e., if $\deltam=0$), then the First Law indicates that the work done \emph{by} that body during the process is identical to the net heat received, $\scrq(\Sigma)$. Thus, if $\scrq(\Sigma)$ is positive then the body \emph{does} work. If $\scrq(\Sigma) = 0$ then the body does no work, nor is any work done on it. 
\end{rem}
\smallskip
The picture usually invoked for a cyclic process is that of a device (e.g., an engine or a refrigerator) in which there are nontrivial temporal variations in the condition of the body suffering the process, with its final condition restored to what it was at the beginning.
	
	This will be the cyclic-process picture that we will often have in mind. However, there is another one, more closely connected with common physical experience, that  will help motivate the mathematical expression of heat transfer from hot to cold. This other picture is that of a body in a \emph{steady} condition, such as the one envisioned in time-invariant solutions of the familiar heat conduction equation, with heat flux at the body's boundary. Because there is no change of condition over time, $\Delta\scrm = 0$ so that over any fixed time interval the process is, in our sense, cyclic.
	
\begin{example}{\emph{(A Simple Steady Heat Transfer Process)}}
\label{ex:SimpleRodHeatTransfer}
 Imagine a very slim cylindrical tube filled with a sample of the material under consideration, a material having state space $\Sigma$. The tube is insulated along its extent but uninsulated at its ends.  The two ends are immersed in different environments that ensure the permanent and laterally uniform presence of state $\sigma'$ at one end and state $\sigma$ at the other. In the picture envisioned, the sample exhibits no change over time (i.e, the sample's condition $\scrm \in \MSigmaPl$  is time-invariant),  with steady heat inflow at the $\sigma'$ end and steady heat outflow at the $\sigma$ end, the rates of heat flow at both ends being identical. Because the sample has a steady condition and there is no \emph{net} heat receipt, the sample experiences no work.

	In the context of this picture we can associate, with a fixed time interval, a cyclic process $\process \in \VSigma$, with $\deltam = 0$ and $\scrq = \alpha\delta_{\sigma'} -  \alpha\delta_{\sigma}$, where $\delta_{\sigma'}$ and $\delta_{\sigma}$ are Dirac measures in \MSigma\ concentrated at $\sigma'$ and $\sigma$. The positive number $\alpha$ is the amount of heat absorbed at the $\sigma'$ end and emitted at the $\sigma$ end during the stipulated time interval. Note that $\scrq(\Sigma) =0$, which is consistent with the absence of work.
\end{example}

\begin{rem} 
\label{rem:ApproxHeatTrans}
The process considered in Example \ref{ex:SimpleRodHeatTransfer} is an idealized one, for it requires, among other things, means to maintain material at the tube ends in permanent and transversely uniform states $\sigma'$ and $\sigma$, and without any temporal change in the condition of the material along the tube's extent. Although, for a thermodynamical theory \theory, such a process might not actually be represented among members of \scrP\  (the true processes), the idealization might well be a member of  $\scrPhat := \textrm{cl}\,(\Cone(\scrP))$. That is, if   \mbox{$(0,\alpha(\delta_{\sigma'} - \delta_{\sigma}))$} is not among the true processes, it might be approximated arbitrarily closely by true processes (or  multiples of them). In such a case, the presence in $\Cone(\scrP)$\  of those approximations would give the same sense of passage of heat from $\sigma'$ to $\sigma$ as would the idealized example.
\end{rem}
\medskip

	To the extent that the direction of heat transfer in a workless cyclic process should guide our conception of  relative hotness in a theory \theory, the presence in \scrP\ (or \scrPhat) of the process in Example \ref{ex:SimpleRodHeatTransfer} would compel us to say that state $\sigma$ is not hotter than state $\sigma'$. However, we refrain from asserting that $\sigma'$ is hotter than $\sigma$: In view of Remark \ref{rem:ApproxHeatTrans}, it might happen that \scrPhat\  also contains a (reverse) element of the form $(0,\bar{\scrq})$, with $\bar{\scrq} = \alpha\delta_{\sigma} -  \alpha\delta_{\sigma'}$. (See Remark \ref{rem:SameHotnessMotivation} below.) In fact, such a possibility provides a basis for asserting that two different states are of the \emph{same} hotness.  

\begin{definition}
\label{def:SameHotness}
  For a thermodynamical theory \theory, two states $\sigma \in \Sigma$ and $\sigma' \in \Sigma$\, \textbf{are of the same hotness} (denoted  $\sigma \sim \sigma'$) if both $(0, \delta_{\sigma} -  \delta_{\sigma'})$ and $(0, \delta_{\sigma'} -  \delta_{\sigma})$ are members of $\hat{\scrP}$. The equivalence relation $\sim$ induces a partition of $\Sigma$ into equivalence classes called the \textbf{hotness levels} of the thermodynamical theory \theory. We denote by \scrH\  the set of hotness levels induced in $\Sigma$ by \scrP, and we give \scrH\  the quotient topology it inherits from $\Sigma$.
\end{definition}
\begin{rem}
\label{rem:SameHotnessMotivation}
	There is no requirement in Definition \ref{def:SameHotness} that $(0, \delta_{\sigma}- \delta_{\sigma'})$ and \mbox{$(0, \delta_{\sigma'}- \delta_{\sigma})$} be members of \scrP, corresponding to the true processes, only that they be well-approximated by positive multiples of \scrP's members. Such approximations might arise if certain members of  \scrP\  correspond to mathematical encodings of physical processes closely resembling the one described in Example \ref{ex:SimpleRodHeatTransfer},  having small  departures from $\sigma'$ and $\sigma$ at the two cylinder ends, some  giving rise to heat flow in one direction and others inducing heat flow in the opposite direction.
\end{rem}

\begin{rem}[\emph{A Dynamical Variant of Example \ref{ex:SimpleRodHeatTransfer}}] The heat transfer process described in Example \ref{ex:SimpleRodHeatTransfer} was a simple toy one in which the body suffering the process was in a steady condition (although not in what is usually regarded as a thermodynamic equilibrium). The example was intended to provide motivation for the idea,  given in Definition \ref{def:SameHotness}, that two states are of the same hotness. Lest it be thought that an element of the form $(0,\alpha(\delta_{\sigma'} -  \delta_{\sigma}))$, in particular with $\Delta \scrm = 0$, can arise in \scrPhat\  only from consideration of physical processes having a temporally unchanging condition, we provide in Appendix \ref{app:TransientPassiveHeatTransfer} a different and more physically robust example in which such steadiness is  never present.
\end{rem}
\medskip
	The following theorem indicates that, for a Kelvin-Planck theory \theory, not only is it true that two states of the same hotness have the same value on each Clausius-Duhem temperature scale \emph{but also that if they are not distinguished by any such scale, then \scrPhat\ must contain the elements stipulated in Definition \ref{def:SameHotness}}.	 For this, the Hahn-Banach Theorem will play a central role.
\begin{samepage}	
\begin{theorem}
\label{thm:EqualHotness}
 Let $\sigma \in \Sigma$ and $\sigma' \in \Sigma$ be two states of the Kelvin-Planck theory \theory. The following are equivalent:
\begin{enumerate}[(i)]
\item $\sigma$ and $\sigma'$ are of the same hotness.
\item $T(\sigma) = T(\sigma')$ for every Clausius-Duhem temperature scale $T \in \scrT_{CD}$.
\end{enumerate}
\end{theorem}
\end{samepage}
	Two lemmas will facilitate the proof that $(ii)$ implies $(i)$. $\MSigmaPlOne$\ denotes the set of $\scrm \in \MSigmaPl$ such that $\scrm(\Sigma) = 1$.

\begin{lemma}
\label{lemma:FirstBigLem}
 Let \theory\ be a Kelvin-Planck theory, let $(\scrv,\scrw)$ be an element of \VSigma, and let $\scrK(\scrv, \scrw)$ be the convex hull of  $(\scrv,\scrw) \cup (0,\MSigmaPlOne)$; that is, let
\begin{equation}
\label{eq:KvwDef}
\scrK(\scrv, \scrw) := \{\lambda(\scrv,\scrw) + (1-\lambda)(0,\scru) : \lambda \in [0,1], \scru \in \MSigmaPlOne\}.
\end{equation}
If $\scrK(\scrv, \scrw)$ is disjoint from \scrPhat\ then there is for the theory a Clausius-Duhem pair $(\eta,T)$ such that
\begin{equation}
\int_{\Sigma}\eta\, d\scrv - \int_{\Sigma}\frac{d\scrw}{T} < 0.
\end{equation}
\end{lemma}

\begin{proof} Because the sets $\{(\scrv,\scrw)\}$ and $(0\MSigmaPlOne)$ are both convex and compact, the convex hull of their union, $\scrK(\scrv, \scrw)$, is also convex and compact (\cite{choquet1969lectures}, \S 19.5). Moreover, by hypothesis $\scrK(\scrv, \scrw)$ is disjoint from \scrPhat. To get the desired result, we need only repeat the Hahn-Banach separation argument in the proof \cite{feinberg-lavineEntropy1} of Theorem \ref{thm:ExistTempEnt}, $(i) \Rightarrow (ii)$, with $\scrK(\scrv, \scrw)$ replacing $(0,\MSigmaPlOne)$.
\end{proof}

\begin{lemma}
\label{lem:CDOpposite}
 Let \theory\ be a Kelvin-Planck theory for which $(\eta^{\,\circ},T^{\,\circ})$ is a Clausius-Duhem pair, and let $(\scrv,\scrw) \in \VSigma$ be such that
\begin{equation}
\label{eq:intCD0}
\int_{\Sigma}\eta^{\,\circ}\, d\scrv - \int_{\Sigma}\frac{d\scrw}{T^{\,\circ}} = 0.
\end{equation}
If $(\scrv,\scrw)$ is not a member of \scrPhat, then there is another Clausius-Duhem pair $(\eta,T)$  such that
\begin{equation}
\label{eq:CDStrictIneq}
\int_{\Sigma}\eta\, d\scrv - \int_{\Sigma}\frac{d\scrw}{T} < 0.
\end{equation}
\end{lemma}

\begin{proof} We first show that $\scrK(\scrv,\scrw)$ intersects \scrPhat\ at most in $(\scrv,\scrw)$. Note that each element $(\scrv^*,\scrw^*)$ in  $\scrK(\scrv,\scrw)$ is of the form
\begin{equation}
\label{eq:v*w*}
(\scrv^*,\scrw^*) = \lambda (\scrv, \scrw) + (0, [1-\lambda] \scru)
\end{equation}
with $\lambda \in [0,1]$ and $\scru \in \MSigmaPlOne$. For $\lambda < 1$, it follows from \eqref{eq:intCD0}, \eqref{eq:v*w*}, the positivity of \scru, and the positivity of $T^{\circ}$ that
\begin{equation}
\label{eq:WrongWay1}
\int_{\Sigma}\eta^{\,\circ}\, d\scrv^* - \int_{\Sigma}\frac{d\scrw^*}{T^{\,\circ}} < 0, \quad \forall\, (\scrv^*,\scrw*) \in \scrK(\scrv,\scrw).
\end{equation}
Because $(\eta^{\,\circ},T^{\,\circ})$ is a Clausius-Duhem pair for \theory, \eqref{eq:WrongWay1} indicates that no member of $\scrK(\scrv,\scrw)$ can be a member of  \scrPhat, except perhaps for $(\scrv,\scrw)$ itself (i.e., corresponding to $\lambda =1$). Thus if, as in the  hypothesis, $(\scrv,\scrw)$ is not a member of \scrPhat, then  $\scrK(\scrv,\scrw)$ and \scrPhat\ are disjoint. In this case, Lemma \ref{lemma:FirstBigLem} ensures the existence of a Clausius-Duhem pair $(\eta,T)$ such that the (strict) inequality \eqref{eq:CDStrictIneq} is satisfied.
\end{proof}

\begin{proof}[Proof of Theorem \ref{thm:EqualHotness}] First we will show that $(ii)$ implies $(i)$. Suppose, then, that $\sigma$ and $\sigma'$ are not distinguished by any Clausius-Duhem temperature scale. We want to show that both $(0, \delta_{\sigma} -  \delta_{\sigma'})$ and $(0, \delta_{\sigma'} -  \delta_{\sigma})$ are members of $\hat{\scrP}$. Let $(\eta^{\,\circ},T^{\,\circ})$ be a Clausius-Duhem pair for \theory, and let $(\scrv,\scrw) = (0, \delta_{\sigma} -  \delta_{\sigma'})$, Note that
\begin{equation}
\int_{\Sigma}\eta^{\,\circ}\, d\scrv - \int_{\Sigma}\frac{d\scrw}{T^{\,\circ}} = \frac{1}{T^{\circ}(\sigma')}-\frac{1}{T^{\circ}(\sigma)} = 0.
\end{equation}
Now suppose that $(\scrv,\scrw) = (0, \delta_{\sigma} -  \delta_{\sigma'})$ is not a member of \scrPhat. From Lemma \ref{lem:CDOpposite} it follows that there is another Clausius-Duhem pair $(\eta,T)$ such that
\begin{equation}
\int_{\Sigma}\eta\, d\scrv - \int_{\Sigma}\frac{d\scrw}{T} = \frac{1}{T(\sigma')}-\frac{1}{T(\sigma)} < 0.
\end{equation}
This, however, contradicts the supposition that no Clausius-Duhem temperature scale distinguishes $\sigma$ from $\sigma'$. The proof that $(0, \delta_{\sigma'} -  \delta_{\sigma})$ is a member of \scrPhat\ is similar.

	To prove that $(i)$ implies $(ii)$ we suppose that both $(0, \delta_{\sigma'} -  \delta_{\sigma})$ and $(0, \delta_{\sigma} -  \delta_{\sigma'})$ are members of \scrPhat. From Remark \ref{rem:Reversible} it follows that, for every Clausius-Duhem pair, the equality \eqref{eq:MainThmCDEQUAL} obtains, with \process\ taken to be $(0, \delta_{\sigma} -  \delta_{\sigma'})$. From this it follows that $T(\sigma) = T(\sigma')$ for all $T \in \scrT_{CD}$.
\end{proof}
\medskip

	Because, for a particular Clausius-Duhem temperature scale, all states in a hotness level have the same temperature, it makes sense to speak of the ``temperature of a hotness level" relative to that scale. In this sense, Theorem \ref{thm:EqualHotness} enables us to think about temperature as a function of hotness rather than as a function of state.
	
\begin{definition} Let $T:\Sigma \to \RP$ be a Clausius-Duhem temperature scale for a Kelvin-Planck theory with hotness levels \scrH. By $T_{*}: \scrH \to \RP$ we mean the \textbf{Clausius-Duhem temperature scale on \scrH\ induced by \emph{T}} the following way: For $h \in \scrH$ let $\sigma \in \Sigma$ be any member of  $h$; then 
\begin{equation}
T_*(h) := T(\sigma).
\end{equation}
The set of all Clausius-Duhem temperature scales induced on \scrH\  by members of $\scrT_{CD}$ will be denoted by $\scrT_{CD*}$.
\end{definition}

\begin{rem} Before closing this sub-section, we provide for a Kelvin-Planck theory \theory\  a few facts, which, among others, were proved (in a slightly different setting) as Lemma 6.3 in \cite{feinberg1983thermodynamics}:
\begin{enumerate}[(i)]
\item Every $T_{*} \in \scrT_{CD*}$ is continuous.
\item \scrH\  is compact and Hausdorff.
\item Every hotness level, viewed as a subset of $\Sigma$, is compact.
\end{enumerate}
\end{rem}

\subsection[A ``Hotter Than" Relation and Its Reflection in Temperature]{A ``Hotter Than" Relation and Its Reflection in Clausius-Duhem Temperature Scales}

	In this section we will introduce for a Kelvin-Planck theory \theory\ a way of making precise the idea that one hotness level is ``hotter than" a different one. 
	
	We begin by defining a \emph{passive heat transfer} from one hotness level to another. In physical terms, this is a cyclic process, perhaps resembling the one in Example \ref{ex:SimpleRodHeatTransfer}, in which there is no work and in which heat is absorbed only by material in states of identical hotness and is emitted, in equal amount, only by material in states having a different common hotness. By the \emph{support} of a measure $\nu \in \MSigmaPl$, denoted $\supp \nu$, we mean the complement in $\Sigma$ of the largest open set of $\nu$-measure zero.
 
\begin{definition}
\label{def:PassiveHT}
 In a thermodynamical theory \theory, a \textbf{passive heat transfer from hotness level h to hotness level $\mathbf{h}^{\prime}$} is an element \process\ of \scrPhat\   such that $\deltam = 0$ and 
\begin{equation}
\scrq = \mu - \mu^{\prime}, 
\end{equation} 
where $\mu$ and $\mu'$ are members of $\MSigmaPl$ such that
\begin{equation}
\label{eq:passivehtreq}
\supp \mu  \subset h, \quad \supp \mu' \subset h', \quad \mu(h) = \mu'(h') > 0.
\end{equation}
\end{definition}
\medskip

To motivate our definition\footnote{In \cite{feinberg1983thermodynamics} there were four subtly different notions of ``hotter than," each with different consequences.  To the extent they can be compared, Definition \ref{def:HotterThan2}, while different in appearance, is closest in spirit to ``hotter than in the second sense" in \cite{feinberg1983thermodynamics}.}  of \emph{hotter than}, we let \theory\  be a Kelvin-Planck theory with distinct hotness levels $h'$ and $h$. Consider an inventor who believes that, through ingenious design, the set of processes indicated in \scrP\  can be expanded to a larger set \emph{that includes a hitherto unknown passive heat transfer from $h$ to $h'$}. We say that $h'$ is hotter than $h$ \emph{if no such expansion is possible without violating the Kelvin-Planck Second Law.}

\begin{definition}
\label{def:HotterThan2}
 Let \theory\ be a Kelvin-Planck theory with hotness levels $h'$ and $h$. Then \textbf{$\mathbf{h}^{\prime}$ is hotter than h} (denoted $h' \succ h$) if $h'\neq h$ and if a thermodynamical theory $(\Sigma,\scrP^{\dagger})$ violates the Kelvin-Planck Second Law --- that is, $(\Sigma,\scrP^{\,\dagger})$ is not a Kelvin-Planck theory --- whenever $\widehat{\scrP^{\,\dagger}}$ contains \scrP\  and also a passive heat transfer from $h$ to $h'$.
\end{definition}

	The preceding definition makes no mention of temperature. The following theorem asserts, for a Kelvin-Planck theory with hotness levels $h'$ and $h$, that $h'$ is hotter than $h$ in the sense of Definition \ref{def:HotterThan2} \emph{precisely} when $h'$ has a higher temperature than $h$ on \emph{every} Clausius-Duhem temperature scale.

\begin{theorem}
\label{thm:HotterThan2AndTemp}
 Let \theory\ be a Kelvin-Planck theory, and let $\scrT_{CD*}$ be its set of Clausius-Duhem temperature scales (on \scrH). If $h'$ and $h$ are distinct hotness levels, the following are equivalent:
\begin{enumerate}[(i)]
\item $h'$ is hotter than $h$.
\item $T_{*}(h') > T_{*}(h),\quad \forall T_{*} \in \scrT_{CD*}$.  
\end{enumerate}
\end{theorem}

\begin{proof} To show the equivalence of (i) and (ii) we will prove the equivalence of their negations:
\begin{enumerate}[$(i)^\prime$]
\item $h'$ is not hotter than $h$.
\item For \theory\  there is a Clausius-Duhem temperature scale $\bar{T}_{*}(\cdot)$ on \scrH\  such that  $\bar{T}_{*}(h') \leq \bar{T}_{*}(h)$.
\end{enumerate}

To prove that $\mathrm{(i)}^\prime$ implies  $\mathrm{(ii)}^\prime$ we first note that $\mathrm{(i)}^\prime$ requires the existence of a Kelvin-Planck theory $(\Sigma,\scrP^{\,\dagger})$  in which $\widehat{\scrP^{\,\dagger}}$ contains \scrP\  and also a passive heat transfer from $h$ to $h'$, say $\scrp^{\dagger} = (0, \mu - \mu^{\prime})$, where $\mu$ and $\mu'$ are measures in \MSigmaPl\   that satisfy \eqref{eq:passivehtreq}. From Theorem \ref{thm:ExistTempEnt} and Remark \ref{rem:CD-clP} there is a Clausius-Duhem pair $(\bar{\eta},\bar{T})$ for $(\Sigma,\scrP^{\,\dagger})$ such that
\begin{equation}\label{eq:CDOnPdagger}
\int_{\Sigma}\bar{\eta}\: d \scrv\  \geq\  \int_{\Sigma}\frac{d\scrw}{\bar{T}}, \quad \forall\  (\scrv,\scrw) \in \widehat{\scrP^{\dagger}}.
\end{equation}
In particular, the inequality in \eqref{eq:CDOnPdagger} obtains for all members of \scrPhat, so   $(\bar{\eta},\bar{T})$ is also a Clausius-Duhem pair for the original Kelvin-Planck theory \theory. Now let $\bar{T}_{*}(\cdot)$ be the \theory-Clausius-Duhem temperature scale on \scrH\  induced by  $\bar{T}(\cdot)$. Because $\scrp^{\dagger}$ is a member of $\widehat{\scrP^{\dagger}}$, 
\eqref{eq:CDOnPdagger} requires that

\begin{equation}
0 \geq \int_{\Sigma}\frac{d(\mu - \mu')}{\bar{T}} = \frac{\mu(h)}{\bar{T}_{*}(h)} - \frac{\mu(h')}{\bar{T}_{*}(h')}  = \mu(h)(\frac{1}{\bar{T}_{*}(h)} - \frac{1)}{\bar{T}_{*}(h')}).
\end{equation}
This in turn requires that $\bar{T}_{*}(h') \leq \bar{T}_{*}(h)$.

	To show that $\mathrm{(ii)}^\prime$ implies  $\mathrm{(i)}^\prime$, we will begin by supposing that $(\bar{\eta}(\cdot),\bar{T}(\cdot))$ is, for \theory, a Clausius-Duhem pair on $\Sigma$ that gives rise to the temperature scale $\bar{T}_*(\cdot)$ on \scrH\ posited in $\mathrm{(ii)}^\prime$. Moreover, we will let $\scrp^{\dagger} = (0, \mu - \mu^{\prime})$, where $\mu$ and $\mu'$ are members of \MSigmaPl\ satisfying \eqref{eq:passivehtreq}. Let
\begin{equation}
\scrP^{\dagger} := \{\process \in \VSigma:  \int_{\Sigma}\bar{\eta\,}d\,\Delta\scrm \geq \int_{\Sigma}\frac{d\scrq}{\bar{T}}\}.
\end{equation}
Then $(\Sigma,\scrP^{\dagger})$ is a Kelvin-Planck theory, with $\widehat{\scrP^{\dagger}}$ containing both \scrP\ and $\scrp^{\dagger}$. From this $\mathrm{(i)}^\prime$ follows.
\end{proof}

	The following is a corollary of Theorem \ref{thm:HotterThan2AndTemp}.

\begin{corollary} For a Kelvin-Planck theory with hotness levels \scrH, the hotter than relation $\succ$ gives a strict partial order on \scrH.
\end{corollary}

\begin{proof} Antisymmetry and transitivity are consequences of \emph{(ii)} in Theorem \ref{thm:HotterThan2AndTemp}.
\end{proof}

	Of course, it might happen that, for a particular Kelvin-Planck theory, \scrH \ is totally ordered by $\succ$.
\smallskip

\begin{rem}[\emph{Totally ordered hotness levels}] It is not generally true that the elements of a set endowed with a total order can be numbered (with real numbers) in such a way as to reflect that order faithfully. A contrived counter-example with a thermodynamic flavor is given in \cite{feinberg1983thermodynamics}. However, an argument along the lines of the proof of Theorem 8.3 in \cite{feinberg1983thermodynamics} gives the following information: \emph{If \scrH\  is the set of hotness levels for a Kelvin-Planck theory and if \scrH\  is totally ordered by $\succ$, then \scrH\ is homeomorphic and order-similar to a subset of the real line. In fact, every $T_* \in \scrT_{CD*}$ reflects the order precisely and provides a homeomorphism between \scrH\  and $T_*(\scrH)$. If $\Sigma$ is connected, then \scrH\  is homeomorphic and order-similar to a closed and bounded interval of the real line.}
\end{rem}

\medskip

	For a Kelvin-Planck theory \theory\  the ordering of the hotness levels in $\Sigma$ by ``$\succ$" can be adapted in an obvious way to make sense of the idea that one state in $\Sigma$ is hotter than another.
	
\begin{definition}\label{def:HotterThanStates} Let \theory\ be a Kelvin-Planck theory, and let $\sigma'$ and $\sigma$ be members of $\Sigma$. Then \textbf{state  $\sigma'$ is hotter than state $\sigma$}, denoted $\sigma' \succ \sigma$, if the hotness level containing $\sigma'$ is hotter than the hotness level containing $\sigma$.
\end{definition}  

The following is an easy corollary of Theorem \ref{thm:HotterThan2AndTemp}.

\begin{corollary}\label{cor:HotterThanStates}
 Let \theory\ be a Kelvin-Planck theory, and let $\scrT_{CD}$ be its set of Clausius-Duhem temperature scales (on $\Sigma$). If $\sigma'$ and $\sigma$ are states in $\Sigma$, the following are equivalent:
\begin{enumerate}[(i)]
\item $\sigma'$ is hotter than $\sigma$.
\item $T(\sigma') > T(\sigma),\quad \forall T \in \scrT_{CD}$.  
\end{enumerate}
\end{corollary}

\subsection{Remarks on Variants of the ``Hotter Than" Relation.}

	To the extent that it is \emph{precisely} reflected in \emph{all} Clausius-Duhem temperature scales for a Kelvin-Planck theory (and because it resonates with the Clausius statement of the Second Law), the definition of \emph{hotter than} ($\succ$)  in the preceding subsection is an especially satisfying one.
	
	There is, however, a weaker but more tangible notion of \emph{hotter than} that can sometimes give information when $\succ$ does not, in particular when two hotness levels are not $\succ$-comparable. Although the context is different, a very similar analog of this weaker notion is discussed extensively in \cite{feinberg1983thermodynamics},\footnote{This weaker notion of \emph{hotter than} corresponds to ``hotter than in the first sense" ($_1{\succ}$) in \cite{feinberg1983thermodynamics}.}  so only a few remarks are provided here.

\begin{definition}
\label{def:HotterThan1}
 Let \theory\ be a Kelvin-Planck theory with hotness levels $h'$ and $h$. Then \textbf{$\mathbf{h}^{\prime}$ is weakly hotter than h} (denoted $h'\, _w{\succ}\; h$) if $h'\neq h$ and 
\scrPhat\  contains a member $(0,\scrq)$ with \scrq\ of the form
\begin{equation}
\label{eq:WorkingHeatTrans1}
\scrq = \mu^{\prime} - \mu + \nu, 
\end{equation} 
where $\mu'$, $\mu$, and $\nu$ are members of $\MSigmaPl$ such that
\begin{equation}
\label{eq:WorkingHeatTrans2}
\supp \mu'  \subset h', \quad \supp \mu \subset h, \quad \mu'(h') = \mu(h') > 0.
\end{equation}
(Here $\nu$ can be the zero measure.)  
\end{definition}

\medspace
	Viewed as a process, $(0,\scrq)$ is a cyclic one in which, over the course of the process, all heat \emph{emitted} from the body suffering the process emanates from material of hotness $h$ and in which there is at least as much heat \emph{absorbed} by material of hotness $h'$. The work done by the body suffering the process, $\scrq(\Sigma) = \nu(\Sigma)$, is not negative. (If $\nu=0$ and $(0,\scrq)$ is a member of \scrPhat\ then the process is a passive heat transfer from $h'$ to $h$, and no work is done.) 
	
	Proof of the following theorem is essentially the one given in \cite{feinberg1983thermodynamics}.
\bigskip

\begin{theorem}
\label{thm:HotterThan1AndTemp}
\samepage{
 Let \theory\ be a Kelvin-Planck theory, and let $\scrT_{CD*}$ be its set of Clausius-Duhem temperature scales (on \scrH). If $h'$ and $h$ are distinct hotness levels, the following are equivalent:
\begin{enumerate}[(i)]
\item $h'$ is weakly hotter than $h$.
\item $T_{*}(h') \geq T_{*}(h),\quad \forall T_{*} \in \scrT_{CD*}$.  
\end{enumerate}
}
\end{theorem}

\begin{rem} When $h'$ and $h$ constitute a \emph{fixed} pair of hotness levels with $h'\, _w{\succ}\; h$, there must be at least one Clausius-Duhem temperature scale $\bar{T}_*(\cdot)$ on \scrH\  such that  $\bar{T}_*(h') >  \bar{T}_*(h)$, for otherwise the two hotness levels would be identical (Theorem \ref{thm:EqualHotness}). In fact, the set of all members of $\scrT_{CD*}$ that distinguish between $h'$ and $h$ is dense in  $\scrT_{CD*}$ (in the sup-norm topology): For if $T_*(\cdot) \in  \scrT_{CD*}$ is such that $T_*(h') = T_*(h)$ then, by choosing $\varepsilon > 0$ sufficiently small, a distinguishing Clausius-Duhem temperature scale $T_*(\cdot) + \varepsilon \bar{T}_*(\cdot)$ can be made to lie any given neighborhood of $T_*(\cdot)$.   

This, however, leaves open the question of whether there is a \emph{single} Clausius-Duhem temperature scale $T^{\circ}_*(\cdot)$ such that    $T^{\circ}_*(h') > T^{\circ}_*(h)$ for \emph{every} pair of hotness levels with $h'\, _w{\succ}\; h$. Such a temperature scale will indeed exist so long as $\Sigma$ has a countable base of open sets, in particular if it is a metric space \cite{feinberg1983thermodynamics}.
\end{rem}
\medskip
\begin{rem}[\emph{An important consequence of Theorems \ref{thm:HotterThan2AndTemp} and \ref{thm:HotterThan1AndTemp} taken together}] Consider a Kelvin-Planck theory \theory \ with Clausius-Duhem temperature scales $\scrT_{CD*}$. Moreover, suppose that $h'$ and $h$ are hotness levels such that $T_*(h') > T_*(h)$ for all $T_*(\cdot)$ in  $\scrT_{CD*}$. Then Theorem \ref{thm:HotterThan2AndTemp} \emph{prohibits the existence} of a passive heat transfer from $h$ to $h'$, while Theorem  \ref{thm:HotterThan1AndTemp} \emph{requires the existence} in \scrPhat\  either a passive heat transfer from $h'$ to $h$ or, more generally, a member $(0, \scrq)$ of the form given by equations \eqref{eq:WorkingHeatTrans1} and \eqref{eq:WorkingHeatTrans2} (roughly, a cyclic-process transfer of heat from $h'$ to $h$ in which the body experiencing the process has no net work done on it).
\end{rem}

\begin{rem} If, in Definition \ref{def:HotterThan1}, \scrq\ can be chosen such that $\nu \neq 0$, then we say that $h'$ is \emph{strongly hotter than} $h$ (denoted $h'\, _s{\succ}\; h$).\footnote{Relation $_s{\succ}$ is the analog of $_3{\succ}$ in \cite{feinberg1983thermodynamics}.} In this case, $h'\, _s{\succ}\; h$ implies that $T_{*}(h') > T_{*}(h)$ for all $T_{*} \in \scrT_{CD*}$. The converse, however, is not generally true: Even when $h'$ has a higher temperature than $h$ on every Clausius-Duhem temperature scale, $h'$ and $h$ might not be $_s{\succ}$-comparable. Nevertheless, they will be $\succ$-comparable by virtue of Theorem \ref{thm:HotterThan2AndTemp}. (See Appendix \ref{app:ClausiusVsCD}, Example \ref{ex:ClScaleEx1}, in which there are only two states.) In any case, we have the implications $h'\, _s{\succ}\; h \; \Rightarrow \; h'\, {\succ}\; h \; \Rightarrow \; h'\, _w{\succ}\; h$.
\end{rem}

\section[Temperature Scale Uniqueness in Kelvin-Planck Theories]{Thermodynamic Temperature Scale Uniqueness in Kelvin-Planck Theories}
\label{sec:TempUniqueness}

	In this section our interest is in the precise connection between the supply of processes  a Kelvin-Planck theory contains and the essential uniqueness of a Clausius-Duhem temperature scale for the theory, first on the entire state space and then on sub-domains of it. Recall from Remark \ref{rem:CDPairNonUnique} that a Kelvin-Planck theory is said to have an \emph{essentially unique} Clausius-Duhem temperature scale if every such scale for the theory is a positive multiple of some fixed one.
	
	Classical arguments appearing in standard textbooks indicate that, if a thermodynamical theory is suitably well endowed with Carnot cycles, then essential uniqueness of a thermodynamic temperature scale is ensured. In the context of Kelvin-Planck theories, we will prove not only this \emph{but also the converse}: Essential uniqueness of a Clausius-Duhem temperature scale for a Kelvin-Planck theory \theory\ \emph{requires} that \scrPhat\ be suitably rich in what we shall call Carnot elements. Here again the proof relies on the Hahn-Banach Theorem in the form of Lemma \ref{lem:CDOpposite}.

\medskip 

	The following definition, already implicit in much that has already been said, will enable us to include in our discussion special ``idealized processes" that, although not among the actual processes represented in \scrP, might nevertheless be members of \scrPhat, in which case they are well-approximated by the actual processes (or by positive multiples of them).

\begin{definition}\label{def:reversible} A \textbf{reversible element} of a thermodynamical theory \theory\ is a member of \scrPhat\  such that its negative is also a member of \scrPhat.   An \textbf{irreversible element}  is a member of \scrPhat\  that is not reversible.\footnote{In particular, an \emph{irreversible process} is a member of \scrP\  that is not a reversible element of the theory.}  A \textbf{cyclic element} of \theory\ is a member of \scrPhat\ of the form $(0,\scrq)$.
\end{definition}

\begin{rem} In classical thermodynamics a reversible process is generally regarded to be one for which there is an associated ``path" that can be reversed in every detail along the path. In Definition \ref{def:reversible} there is no such insistence on detailed path reversal;  there is only the requirement that both \process\ and its negative be members of \scrPhat.
\end{rem}

\begin{definition}\label{def:CarnotElem} Let \theory\ be a thermodynamical theory. A \textbf{Carnot element} of the theory is a reversible cyclic element $(0, \scrq) \in \scrPhat$, with \scrq\ having a representation  of the following kind: There are hotness levels $h'$ and $h$ such that
\begin{equation}
\scrq = \mu' - \mu,
\end{equation}
where $\mu'$ and $\mu$ are non-zero measures in $\MSigmaPl$ satisfying
\begin{equation}
\supp\, \mu' \subset h' \quad\mathrm{and}\quad \supp\, \mu \subset h.
\end{equation}
In this case, the Carnot element  \textbf{operates between hotness levels ${\mathbf{h'}}$ and $\mathbf{h}$}. In the special case that $\scrq = c'\delta_{\sigma'} -  c\delta_{\sigma}$ where $c'$ and $c$ are positive constants and $\sigma'$ and $\sigma$ are members of $\Sigma$ we say that the Carnot element \textbf{operates between states $\sigma'$ and $\sigma$}.
\end{definition}

\begin{rem}\label{rem:CarnotCycles2States}
 Regarded in terms of the usual textbook picture, a Carnot element operating between states $\sigma'$ and $\sigma$ can be thought of as encoding the limit of extremely narrow Carnot cycles having two minuscule isothermal segments centered on $\sigma'$ and $\sigma$.
\end{rem}

\subsection[Temperature Uniqueness: The Inexorable Role of Carnot Cycles]{Essential Uniqueness of a Thermodynamic Temperature Scale: The Inexorable Role of Carnot Elements}

 In standard textbook arguments, the (tacitly assumed) presence of a large supply of Carnot cycles ensures not only the existence of a thermodynamic temperature scale but also its  \emph{essential uniqueness}.  With respect to uniqueness, Theorem \ref{thm:UniquenessTemp}\footnote{This theorem resembles Theorem 9.1 in \cite{feinberg1983thermodynamics}, but, in the broader setting of this article, a \emph{Clausius-Duhem} temperature scale has a meaning somewhat different from a \emph{Clausius} temperature scale in \cite{feinberg1983thermodynamics}. See Appendix \ref{app:ClausiusVsCD}.} also asserts the converse: For a Kelvin-Planck theory \theory\   to have an essentially unique Clausius-Duhem temperature scale on $\Sigma$, it is \emph{necessary} that the theory be so rich in Carnot elements that there is at least one operating between each pair of hotness levels. For proof of this converse, the Hahn-Banach theorem plays a critical role, again in the guise of  Lemma \ref{lem:CDOpposite}. 
\smallskip 

\begin{samepage}
\begin{theorem}\label{thm:UniquenessTemp}  Let \theory\ be a Kelvin-Planck theory with hotness levels \scrH, and let $T(\cdot)$ be a Clausius-Duhem temperature scale on $\Sigma$.  The following are equivalent:
\begin{enumerate}[(i)]
\item Every Clausius-Duhem temperature scale on $\Sigma$ is a positive multiple of $T(\cdot)$.
\item If \scrq\ is a member of $\scrM(\Sigma)$ that satisfies
\begin{equation}\label{eq:q/TIntegrates to zero}
\int_{\Sigma}\frac{d\scrq}{T} = 0
\end{equation}
then $(0,\scrq)$ is a member of \scrPhat.  
\item For each pair of hotness levels $h' \in \scrH$ and $h \in \scrH$ there is a Carnot element of \theory\  operating between $h'$ and $h$.
\item For each pair of states $\sigma' \in \Sigma$ and $\sigma \in \Sigma$ there is a Carnot element of \theory\ operating between them,  having the form $(0, \scrq)$ with
\begin{equation}\label{eq:qInTwoStateCarnotCycle}
\scrq = c' \,\delta_{\sigma'} -  c\,\delta_{\sigma} \quad \mathrm{and} \quad \frac{c'}{c} = \frac{T(\sigma')}{T(\sigma)}.
\end{equation} 
\end{enumerate}
\end{theorem}
\end{samepage}

\begin{rem}[\emph{Existence vs. Uniqueness, 1}]\label{rem:ExistVsUnique1} In the companion article \cite{feinberg-lavineEntropy1}, it was shown that for any Kelvin-Planck theory there invariably exists a pair of continuous functions on the state space, a specific entropy and a thermodynamic temperature, that complies with the Clausius-Duhem inequality for all processes the theory contains. The \emph{existence} of such a pair followed directly from the Hahn-Banach Theorem and did not require the presence in the theory of special processes such as Carnot cycles or, more generally, reversible processes. 

However, for the essential \emph{uniqueness} of the thermodynamic temperature function, the equivalence of $(i)$ and $(iv)$ in Theorem \ref{thm:UniquenessTemp} indicates that every state should be ``visited" by a Carnot cycle, in particular by a reversible process. Some readers might infer from this that, for the essential \emph{uniqueness} of a thermodynamic temperature scale for given Kelvin-Planck theory,  every member of that theory's state space should, in some sense, be an ``equilibrium" state. This is discussed further in Section \ref{sec:ConclRems,Part II}.
\end{rem}

\begin{proof}[Proof of Theorem \ref{thm:UniquenessTemp}]
We will first show that $(i) - (iii)$ are equivalent, and then we will show that (iv) is equivalent to these.

Proof that $(i)$ implies $(ii)$ is a straightforward application of Lemma \ref{lem:CDOpposite}. To prove that $(ii)$ implies $(iii)$, let $T_*(\cdot)$ be the temperature scale on \scrH\  induced by $T(\cdot$), and let $h'$ and $h$ be hotness levels. Furthermore, let $\scrq$ be a measure\footnote{An example is given by \eqref{eq:qInTwoStateCarnotCycle}, where $\sigma'$ and $\sigma$ are states in $h'$ and $h$.} in  $\scrM(\Sigma)$ of the form 
\begin{equation}
\scrq = \mu' - \mu,
\end{equation}
where $\mu'$ and  $\mu$ are nonzero measures in $\scrM_+(\Sigma)$ having support in $h'$ and $h$ respectively and satisfying the equation
\begin{equation}
\frac{\mu(h')}{\mu(h)} = \frac{T_*(h')}{T_*(h)}.
\end{equation}
 Note that \scrq, so chosen, satisfies \eqref{eq:q/TIntegrates to zero}, as does its negative, so $(ii)$ ensures that $(0,q)$ and $(0,-q)$ are both member of \scrPhat.  In fact, from its form, $(0,q)$ is a Carnot element operating between $h'$ and $h$.

To prove that $(iii)$ implies $(i)$ we suppose that $\bar{T}(\cdot)$ is another Clausius-Duhem temperature scale on $\Sigma$, different from $T(\cdot)$, and that $\sigma_0$ is some fixed state in $\Sigma$.  Our aim will be to show that
\begin{equation}
\bar{T}(\cdot) = \frac{\bar{T}(\sigma_0)}{T(\sigma_0)}T(\cdot).
\end{equation} 
For this purpose, let $\sigma$ be an arbitrary state and let $h$ and $h_0$ be the hotness levels containing $\sigma$ and $\sigma_0$. From $(iii)$ there is a Carnot element operating between $h$ and $h_0$. This is to say that \scrPhat\ contains a reversible element $(0, \mu - \mu_0)$
where $\mu$ and $\mu_0$ are non-zero measures in $\MSigmaPl$ satisfying
\begin{equation}
\supp\, \mu \subset h \quad\mathrm{and}\quad \supp\, \mu_0 \subset h_0.
\end{equation}
If $\bar{T}_*(\cdot)$ and $T_*(\cdot)$ are the temperature scales on \scrH\ corresponding to $\bar{T}(\cdot)$ and $T(\cdot)$, we can invoke the Clausius-Duhem inequality to write, for the reversible element $(0, \mu - \mu_0) \in \scrPhat$,
\begin{equation}
\frac{\mu(h)}{T_*(h)} - \frac{\mu_0(h_0)}{T_*(h_0)} = 0 \quad \textrm{and} \quad \frac{\mu(h)}{\bar{T}_*(h)} - \frac{\mu_0(h_0)}{\bar{T}_*(h_0)} = 0.
\end{equation}
From this it follows that
\begin{equation}
\bar{T_*}(h) = \frac{\bar{T_*}(h_0)}{T_*(h_0)}T_*(h).
\end{equation}
Since $\sigma \in \Sigma$ was arbitrary, we actually have
\begin{equation}
\bar{T}(\sigma) = \frac{\bar{T}(\sigma_0)}{T(\sigma_0)}T(\sigma), \quad \forall\ \sigma \in \Sigma,
\end{equation}
which is what $(i)$ asserts.

Having shown that $(i) - (iii)$ are equivalent, we will now show that these are equivalent to $(iv)$. It is easy to see that $(iv)$ implies $(iii)$. Next we show that $(ii)$ implies $(iv)$. Because \scrq\ given by \eqref{eq:qInTwoStateCarnotCycle} satisfies \eqref{eq:q/TIntegrates to zero}, $(ii)$ ensures that both $(0,\scrq)$ and its negative are members of \scrPhat. In this case, $(iv)$ is satisfied.
\end{proof}
\medskip

	We conclude this subsection with statements of two corollaries of Theorem \ref{thm:UniquenessTemp}, proofs of which (omitted here) are very much like the proofs of Corollaries 9.1 - 9.3 in  \cite{feinberg1983thermodynamics}, although the context there is different (Remark \ref{rem:ClausVsCDTempScales}).  Moreover, the proof of Corollary \ref{cor:UniqueTHalfSpaceCor}  resembles the proof given in the next section of the substantially broader Corollary \ref{cor:CDuniqueness&irreversibility}.

\bigskip
\begin{samepage}
\begin{corollary}\label{cor:UniqueTHalfSpaceCor}
 Let \theory\ be a Kelvin-Planck theory for which all Clausius-Duhem temperature scales on $\Sigma$ are positive multiples of some fixed one, $T(\cdot)$. The reversible cyclic elements of \theory\ are precisely those elements $(0,\scrq) \in \VSigma$ that satisfy
\begin{equation}\label{eq:q/TIntegrates to zero2}
\int_{\Sigma}\frac{d\scrq}{T} = 0.
\end{equation}
In particular, any $(0,\scrq) \in \VSigma$ that satisfies \eqref{eq:q/TIntegrates to zero2} is a member of \scrPhat. Of those $(0,\scrq) \in \VSigma$ that satisfy
\begin{equation}
\int_{\Sigma}\frac{d\scrq}{T} < 0,
\end{equation}
either all are contained in \scrPhat\ or none are. 
\end{corollary}
\end{samepage}
	A  component of standard textbook arguments underlying the foundations of classical thermodynamics, in  particular the existence of an entropy as a function of state, relies on an assertion to the effect that any cyclic reversible process can be approximated by combinations of Carnot cycles. (See, for example, page 35 in \cite{denbigh1981principles}.) The following ensures that, for any Kelvin-Planck theory in which there is an essentially unique Clausius-Duhem temperature scale, the supply of Carnot elements is sufficiently large as to make that assertion true. 

\begin{corollary}[Approximating reversible cyclic elements by combinations of Carnot elements] Let \theory\ be a Kelvin-Planck theory for which all Clausius-Duhem temperature scales on $\Sigma$ are positive multiples of some fixed one. The set of all linear combinations of Carnot elements of \theory\ is dense in the set of all reversible cyclic elements of \theory.
\end{corollary}

\subsection[Temperature Scale Uniqueness on a State Space Sub-domain]{Essential Uniqueness of a Thermodynamic Temperature on a State Space Sub-domain}

	Although a thermodynamic temperature scale for a Kelvin-Planck theory need not be essentially unique on the entire state space, there might be nontrivial sub-domains on which essential uniqueness is to be found. The connection to Carnot elements (and perhaps to notions of equilibrium states) is recorded in the following proposition.

\begin{proposition} \label{prop:SubdomainTempUniq}

 Let \theory\ be a Kelvin-Planck theory having $\scrT_{CD}$ as its set of  Clausius-Duhem temperature scales . Furthermore, let $\Sigma^{\,0}$ be a subset of $\Sigma$, and let $\scrT^{\,0}_{CD}$ be the set of restrictions of members of $\scrT_{CD}$ to $\Sigma^{\,0}$ . The following are equivalent:
\begin{enumerate}[(i)]
\item All member of $\scrT^{\,0}_{CD}$ are positive multiples of some fixed one.
\item For each pair of distinct states in $\Sigma^{\,0}$ there is a Carnot element operating between them.
\end{enumerate}
\end{proposition}

\begin{proof}
That (ii) implies (i) is a direct consequence of the Clausius-Duhem inequality. Proof that (i) implies (ii) amounts to an application of Lemma \ref{lem:CDOpposite}.
\end{proof}

\begin{rem}\label{CarnotEquiv} For a Kelvin-Planck theory \theory\ we will say that states $\sigma$ and $\sigma'$ in $\Sigma$ are \emph{Carnot-related}, denoted $\sigma \approx_{\,\scrC}\sigma'$, if $\sigma = \sigma'$ or if  \scrPhat\  contains a Carnot element in operating between them. It is not difficult to see that $\approx_{\,\scrC}$ is an equivalence relation. Proposition \ref{prop:SubdomainTempUniq} tells us that we have essential temperature-scale uniqueness on each nontrivial $\approx_{\,\scrC}$ equivalence class.
\end{rem}

\section{Uniqueness of Entropy-Temperature  Functions of State}
\label{sec:EntropyUniqueness}
	In this section we ask: For a Kelvin-Planck theory \theory, what must be true of \scrP\  beyond a rich supply of Carnot elements to ensure not only that there is an essentially unique Clausius-Duhem temperature scale but also that there be an essentially unique Clausius-Duhem specific-entropy function?

\subsection{Entropy-Temperature-Pair Uniqueness on the Entire State Space}

	The following theorem describes conditions under which, for a Kelvin-Planck theory, there is an essentially unique Clausius-Duhem pair on the entire state space.
\begin{samepage}
\begin{theorem}[Clausius-Duhem Pair Uniqueness]\label{thm:CDPairUniqueness}
 Let $(\etao,\To)$ be a Clausius-Duhem pair for a thermodynamical theory \theory. The following are equivalent:
\begin{enumerate}[(i)]
\item If $(\eta,T)$ is any other Clausius-Duhem pair for \theory, there are constants $\alpha$ and $\beta$, with $\alpha > 0$, such that
\begin{equation}
T(\cdot) = \alpha\To(\cdot)\quad and \quad \eta(\cdot) = \frac{1}{\alpha}\etao(\cdot)
 + \beta. 
\end{equation}
\item \scrPhat\  contains the hyperplane
\begin{equation} \label{eq:hyperplane}
\{\  \process \in \VSigma: \int_{\Sigma}\etao\, d\,(\Delta\scrm) = \int_{\Sigma}\frac{d\scrq}{
\To}\ \}.
\end{equation}

\item For each choice of $\sigma'$ and $\sigma$ in $\Sigma$, \scrPhat\   contains a reversible element with change of condition $\delta_{\sigma'} - \delta_{\sigma}$ and also a Carnot element operating between the hotness levels of $\sigma'$ and $\sigma$.

\end{enumerate}
\end{theorem}
\end{samepage}

\begin{rem}[\emph{Existence vs. Uniqueness, 2}] The comments made in Remark \ref{rem:ExistVsUnique1} pertain here too. Although the \emph{existence} of a Clausius-Duhem entropy-temperature pair for a Kelvin-Planck theory is ensured immediately by the Hahn-Banach Theorem without the requirement of special reversible processes \cite{feinberg-lavineEntropy1}, the equivalence of $(i)$ and $(iii)$ in Theorem \ref{thm:CDPairUniqueness} indicates that, for a Kelvin-Planck theory to have an essentially \emph{unique} Clausius-Duhem pair, not only must every state be visited by Carnot cycles, each must also be visited by other reversible processes as well. As in Remark \ref{rem:ExistVsUnique1}, some readers might infer that all states in such a Kelvin-Planck theory must be ``equilibrium" states. Again, this is discussed on Section \ref{sec:ConclRems,Part II}.
\end{rem}

\begin{proof}[Proof of Theorem \ref{thm:CDPairUniqueness}] To prove that $(i)$ implies $(ii)$ we suppose that $(i)$ holds and that  there exists $ \process \in \VSigma$ that satisfies the equation in \eqref{eq:hyperplane} but does belong to \scrPhat. Then Lemma \ref{lem:CDOpposite} ensures the existence of another Clausius-Duhem pair $(\eta, T)$ such that
\begin{equation}
\int_{\Sigma}\eta\; d\,(\Delta\scrm) < \int_{\Sigma}\frac{d\scrq}{
T}.
\end{equation}
Because \process\ satisfies the equation in \eqref{eq:hyperplane}, it is evident that $\eta(\cdot)$ and $T(\cdot)$ could not be of the form given in $(i)$. Thus, we have a contradiction.

	To prove that $(ii)$ implies $(iii)$, for an arbitrary choice of $\sigma'$ and $\sigma$ in $\Sigma$  we first let $\scrq$ be any member of \MSigma\ that satisfies the equation
\begin{equation}
\etao(\sigma') - \etao(\sigma) = \int_{\Sigma}\frac{d\,\scrq}{\To}.
\end{equation}
Then $(\delta_{\sigma'} - \delta_{\sigma},\scrq) \in \VSigma$ and its negative both satisfy the equation in \eqref{eq:hyperplane}, so from $(ii)$ both are members of \scrPhat. Finally, note that $(0,\scrq^*) \in \VSigma$, where $\scrq^*$ is any member of \MSigma\  of the form 
\begin{equation}\label{eq:qInTwoStateCarnotCycle2}
\scrq^* = c' \,\delta_{\sigma'} -  c\,\delta_{\sigma} \quad \mathrm{with} \quad \frac{c'}{c} = \frac{\To(\sigma')}{\To(\sigma)}.
\end{equation}
satisfies the equation in \eqref{eq:hyperplane}, as does its negative. From $(ii)$, then, both are members of \scrPhat, so $(0,\scrq^*)$ is the desired Carnot element. 

	We turn next to a proof that $(iii)$ implies $(i)$. When $(iii)$ holds it is evident that for any pair of hotness levels there is a Carnot element operating between them. From Theorem \ref{thm:UniquenessTemp}, if $T(\cdot)$ is a Clausius-Duhem temperature scale on $\Sigma$ we already have the existence of a positive $\alpha$ such that $T(\cdot) = \alpha \To(\cdot)$. 
	
	If $\eta$ is a specific-entropy function corresponding to $T$, it remains to be shown that, when $(iii)$ holds, $\eta$ is of the form given in $(i)$. Let $\sigma^* \in \Sigma$ be some fixed state, and let $\sigma$ be some other arbitrary state, Then from $(iii)$ there is in \scrPhat\  a reversible element of the form $(\delta_{\sigma} - \delta_{\sigma^*},\scrq)$, which must be consistent with the Clausius-Duhem inequality written in terms of both Clausius-Duhem pairs ($(\eta,T)$ and $(\etao,\To)$.  Thus, we have, for all choices of $\sigma$, 
\begin{equation}
\eta(\sigma) - \eta(\sigma^*) = \int_{\Sigma}\frac{d\scrq}{T} = \frac{1}{\alpha}\int_{\Sigma}\frac{d\scrq}{\To}  
\end{equation}
and
\begin{equation}
\etao(\sigma) - \etao(\sigma^*) = \int_{\Sigma}\frac{d\scrq}{\To}.
\end{equation}
Because $\sigma$ was arbitrary, it follows from these equations that
\begin{equation}
\eta(\cdot) = \frac{1}{\alpha}\etao(\cdot) + [\eta(\sigma^*) - \frac{1}{\alpha}\etao(\sigma^*)],
\end{equation}
which is in the form required by $(i)$.
\end{proof}

\begin{corollary}
\label{cor:CDuniqueness&irreversibility}
 Consider a Kelvin-Planck theory \theory\  having, in the sense of $(i)$, an essentially unique Clausius-Duhem pair, $(\etao,\To)$. The set of reversible elements in \scrPhat\  coincides with the set of all $\process \in \VSigma$\  that satisfy
\begin{equation} \label{eq:CDEquality}
 \int_{\Sigma}\etao\, d\,(\Delta\scrm) = \int_{\Sigma}\frac{d\scrq}{
\To}.
\end{equation} 
If even one member of the set 
\begin{equation} \label{eq:CDOpenHalf-space}
\{\  \process \in \VSigma: \int_{\Sigma}\etao\, d\,(\Delta\scrm) > \int_{\Sigma}\frac{d\scrq}{
\To}\ \}
\end{equation}
is an element of \scrPhat\  then all are. In particular, if \scrP\ contains even one irreversible process, then \scrP\ is so rich in processes that \scrPhat\  contains all members of \VSigma\  that are consistent with the (not necessarily strict) Clausius-Duhem inequality.
\end{corollary}

\begin{rem} In the context of Corollary \ref{cor:CDuniqueness&irreversibility}, if \scrPhat\ contains even one irreversible element, \scrPhat\ actually coincides with the closed half-space in \VSigma\ given by
\begin{equation} \label{eq:CDClosedHalf-space}
\{\  \process \in \VSigma: \int_{\Sigma}\etao\, d\,(\Delta\scrm) \geq \int_{\Sigma}\frac{d\scrq}{
\To}\ \}.
\end{equation}
\end{rem}
\medskip
\begin{proof}[Proof of Corollary \ref{cor:CDuniqueness&irreversibility}] 

	If $\process \in \VSigma$\  satisfies \eqref{eq:CDEquality} then so does its negative, in which case Theorem \ref{thm:CDPairUniqueness} requires that both be members of \scrPhat. Hence, \process\ is a reversible element of \scrPhat. On the other hand, if \process\ is a reversible element of \scrPhat\  then both it and its negative are members of \scrPhat. Because each is a member of \scrPhat, both must satisfy the Clausius-Duhem inequality, so the equality \eqref{eq:CDEquality} must obtain.

	To prove the remainder of the corollary, we let \process\ and $(\Delta\scrm^*,\scrq^*)$ be members of \VSigma\ that satisfy
\begin{equation}
\int_{\Sigma}\etao\, d\,(\Delta\scrm) > \int_{\Sigma}\frac{d\scrq}{
\To}\quad\textrm{and}\quad\int_{\Sigma}\etao\, d\,(\Delta\scrm^*) > \int_{\Sigma}\frac{d\scrq^*}{
\To},
\end{equation}
and we suppose that $(\Delta\scrm^*,\scrq^*)$  is a member of \scrPhat. Our aim is to show that \process\ is also a member of \scrPhat.

	Let $\gamma$ denote the positive number defined by
\begin{equation}
\gamma:= \frac{\int_{\Sigma}\etao\, d\,(\Delta\scrm) - \int_{\Sigma}\frac{d\scrq}{
\To}}{\int_{\Sigma}\etao\, d\,(\Delta\scrm^*) - \int_{\Sigma}\frac{d\scrq^*}{
\To}}\,.
\end{equation}
Note that
\begin{equation}\label{eq:TwoTermSum}
\process = \gamma(\Delta\scrm^*,\scrq^*) + [\process - \gamma(\Delta\scrm^*,\scrq^*)].
\end{equation}
Because \scrPhat\  is a cone, the first term on the right is a member of \scrPhat. To see that the second term is also a member of \scrPhat, note that the second term can be rewritten as
\begin{equation}\label{eq:SecondTerm}
(\Delta\scrm - \gamma\Delta\scrm^*\,,\  \scrq - \gamma\scrq^*).
\end{equation}
and that
\begin{align}
&\int_{\Sigma}\etao d(\Delta\scrm - \gamma\Delta\scrm^*) - \int_{\Sigma}\frac{d(\scrq - \gamma\scrq^*)}{\To}\nonumber\\
&= (\int_{\Sigma}\etao d(\Delta\scrm) - \int_{\Sigma}\frac{d\scrq}{\To}) - \gamma(\int_{\Sigma}\etao d(\Delta\scrm^*) - \int_{\Sigma}\frac{d\scrq^*}{\To})\\ 
&= 0\nonumber.
\end{align}

From  the equivalence of $(i)$ and $(ii)$ in Theorem \ref{thm:CDPairUniqueness}, the element \eqref{eq:SecondTerm} --- and therefore the second term on the right of \eqref{eq:TwoTermSum} ---  must be a member of \scrPhat. Because both terms on the right of \eqref{eq:TwoTermSum} are members of the convex cone \scrPhat, their sum \process\  is a member of \scrPhat.

	To prove the last sentence of the corollary, let $(\Delta\scrm^{\dagger},\scrq^{\dagger}) \in \scrP$ be an irreversible process. By the Clausius-Duhem inequality, we must have
\begin{equation} \label{eq:CDo}
 \int_{\Sigma}\etao\, d\,(\Delta\scrm^{\dagger}) \geq \int_{\Sigma}\frac{d\scrq^{\dagger}}{
\To}.
\end{equation}
Because $(\Delta\scrm^{\dagger},\scrq^{\dagger})$ is not reversible, strict inequality must hold in \eqref{eq:CDo}, so \scrPhat\  contains the entire open half-space \eqref{eq:CDOpenHalf-space}. Theorem \ref{thm:CDPairUniqueness} ensures that \scrPhat\ also  contains the hyperplane \eqref{eq:hyperplane}.	 
\end{proof}

\begin{rem}[\emph{Consequences of a change of condition that cannot be reversed}] Suppose that \theory\  is a Kelvin-Planck theory having an essentially unique Clausius-Duhem temperature scale, $T^0$. If there is even one \emph{change of condition} (as distinct from a process) that is not reversible --- that is, if there exists \mbox{$\Delta\scrm^0 \in \MSigmaZ$} such that for no choice of $\scrq$ is $(\Delta\scrm^0,\scrq)$ a reversible element of $\scrPhat$ --- then $T^0$ cannot have an essentially unique  Clausius-Duhem entropy partner. For if $\eta^{\,0}$ were such a partner then, for any $\scrq \in \MSigma$ that satisfies
\begin{equation}
\int_{\Sigma}\eta^0\, d\, \Delta\scrm^0 = \int_{\Sigma}\frac{d\, \scrq}{T^0},
\end{equation}
Corollary \ref{cor:CDuniqueness&irreversibility} would require that $(\Delta\scrm^0,\scrq)$ be a reversible element of \scrPhat.

	Of course the presence of an abundance of irreversible \emph{processes} in a Kelvin-Planck theory \theory\ does not preclude for it an essentially unique Clausius-Duhem entropy-temperature pair, as in Corollary \ref{cor:CDuniqueness&irreversibility} when \scrPhat\ is a half-space of \VSigma.
\end{rem} 

\begin{rem} Consider a Kelvin-Planck theory \theory\ for which the set of Clausius-Duhem entropy-temperature pairs is not essentially unique. If $(\Delta\scrm^0,\scrq^0)$ is an irreversible element of \scrPhat\ then there must exist at least one Clausius-Duhem pair for which the Clausius-Duhem inequality applied to $(\Delta\scrm^0,\scrq^0)$ is strict, for otherwise  $(\Delta\scrm^0,\scrq^0)$ would, by Corollary \ref{cor:CDuniqueness&irreversibility}, be reversible.

	This prompts the following question: Is there a \emph{single} Clausius-Duhem pair with respect to which the Clausius-Duhem inequality is strict when applied to \emph{every} irreversible element of \scrPhat\,? When $\Sigma$ is a metric space the answer is \emph{yes}. This follows from  an argument similar to the one given in the proof of Theorem 7.2 in \cite{feinberg1983thermodynamics}. 
\end{rem}

\subsection{Essential Uniqueness of Entropy on a State-Space Sub-domain}

	Even when, for a Kelvin Planck \theory, there is an essentially unique Clausius-Duhem temperature scale on a state-space sub-domain $\Sigma_0 \subset \Sigma$, we cannot expect in general that there will invariably be an essentially unique specific-entropy function on $\Sigma_0$. The following theorem describes precisely the circumstances under which such entropy-uniqueness on that sub-domain will obtain. Although the implication (ii) $\Rightarrow$ (i) is straightforward, the less obvious reverse implication follows from Hahn-Banach Theorem in the guise of Lemma \ref{lem:CDOpposite}. As might be expected, the situation is very much like that in Theorem \ref{thm:CDPairUniqueness}, but with some subtle differences.

\begin{theorem}
\label{thm:EntropyUniqueness}
 Let \theory\ be a Kelvin-Planck theory, and let $\Sigma_0$ be a subset of\, $\Sigma$ consisting of at least two states. Suppose that $T^{\dagger}_0: \Sigma_0 \to \RP$ is the restriction to $\Sigma_0$ of a Clausius-Duhem temperature scale $T^{\dagger}: \Sigma \to \RP$ and that every other restriction of a Clausius-Duhem temperature scale to $\Sigma_0$ is a positive multiple of $T^{\dagger}_0$. The following are equivalent:
\begin{enumerate}[(i)]
\item If $(\eta,T^{\dagger})$ and $(\bar{\eta},T^{\dagger}))$ are Clausius-Duhem pairs for \theory\ then, restricted to $\Sigma_0$, $\eta$ and $\bar{\eta}$  differ by at most a constant.
\item For each pair of distinct states $\sigma' \in \Sigma_0$ and $\sigma \in \Sigma_0$, there exists in \scrPhat\  a reversible element $(\delta_{\sigma'} - \delta_{\sigma}, \scrq)$, with  the support of \scrq\ contained in $\Sigma_0$.\footnote{The support of the signed measure \scrq\ is the union of the supports of its (Hahn-Jordan) positive and negative parts.} 
\end{enumerate}
\end{theorem}

\begin{proof} To prove that $(i)$ implies $(ii)$ suppose, on the contrary, that $(i)$ holds but that $\sigma'$ and $\sigma$ are states in $\Sigma_0$ such that, for no choice of $\scrq \in \MSigma$ with $\supp \scrq \subset \Sigma_0$, are both $(\delta_{\sigma'} - \delta_{\sigma},\scrq)$ and its negative members of \scrPhat. In particular, $(\delta_{\sigma'} - \delta_{\sigma},\scrq^*)$ and its negative cannot both be members of \scrPhat, where  $\scrq^*$ is chosen to be a member of \MSigma\ that has support in $\Sigma_0$ and that satisfies the equation

\begin{equation}\label{eq:NoBarEqual}
\int_{\Sigma}\eta\ d(\delta_{\sigma'} - \delta_{\sigma}) = \eta\,(\sigma') -  \eta\,(\sigma) = \int_{\Sigma}\frac{d\scrq^*}{T^{\,\dagger}}.
\end{equation}
Here  $(\eta,T^{\dagger})$ is the first Clausius-Duhem pair in $(i)$.

	If $(\delta_{\sigma'} - \delta_{\sigma},\scrq^*)$ is not a member of \scrPhat, then  Lemma \ref{lem:CDOpposite} ensures that there is another Clausius-Duhem pair $(\eta^{\,\#}, T^{\,\#})$ such that
\begin{equation}\label{eq:SharpInequal}
 \eta^{\,\#}\,(\sigma') -  \eta^{\,\#}\,(\sigma) < \int_{\Sigma}\frac{d\scrq^*}{T^{\#}}.
\end{equation}
Recall from Remark \ref{rem:CDPairNonUnique} that, for any $\alpha > 0$, $(\frac{1}{\alpha}\,\eta^{\,\#},\alpha T^{\,\#})$ is again a Clausius-Duhem pair. In particular, from the hypothesis of the theorem, there is an $\alpha^{\,*} > 0$ such that $\alpha^{\,*}T^{\,\#}(\cdot)$ and $T^{\,\dagger}(\cdot)$ are identical on $\Sigma_0$. Thus, with
\begin{equation}
\bar{\eta}\,(\cdot) := \frac{1}{\alpha^*}\eta^{\,\#}(\cdot),
\end{equation}

\smallskip
\noindent
$(\bar{\eta},T^{\dagger})$ is a Clausius-Duhem pair for \theory. From \eqref{eq:SharpInequal} and the fact that $\scrq^*$ has support in $\Sigma_0$ it follows that
\begin{equation}\label{eq:BarInequal}
\bar{\eta}(\sigma') -  \bar{\eta}(\sigma) < \int_{\Sigma}\frac{d\scrq^*}{T^{\,\dagger}}.
\end{equation} 
Comparison with \eqref{eq:NoBarEqual} tells us that the specific entropy functions $\bar{\eta}$ and $\eta$, both corresponding to the temperature scale $T^{\,\dagger}$, cannot differ on $\Sigma_0$ by at most a constant, in contradiction to $(i)$. If $-\,(\delta_{\sigma'} - \delta_{\sigma},\scrq^*)$ is not a member of \scrPhat, proof of contradiction to $(i)$ is similar.

	To prove that $(ii)$ implies $(i)$ suppose that $\bar{\eta}$ and $\eta$ are specific-entropy functions on $\Sigma$ corresponding to the same Clausius-Duhem temperature scale $T^{\dagger}$. Let $\sigma_0$ be a fixed state in  $\Sigma_0$, and let $\sigma \in \Sigma_0$ be another state. From (ii) it follows that \scrPhat\ contains a reversible element $(\delta_{\sigma} - \delta_{\sigma_0},\scrq)$. Because the element is reversible, the Clausius-Duhem inequality requires that 
\begin{equation}
\bar{\eta}(\sigma) - \bar{\eta}(\sigma_0) = \eta\,(\sigma) - \eta\,(\sigma_0) = \int_{\Sigma}\frac{d\scrq}{T^{\,\dagger}}.
\end{equation}
Thus, for any choice of $\sigma \in \Sigma_0$ we have
\begin{equation}
\bar{\eta}(\sigma) - \eta(\sigma) = \eta(\sigma_0) - \eta(\sigma_0),
\end{equation}
which is to say that $\bar{\eta}$ and $\eta$, restricted to $\Sigma_0$, differ by at most a constant.
\end{proof}

\begin{rem} In the proof that $(ii)$ implies $(i)$ there was no need to require that the heating measure have support in $\Sigma_0$; any heating measure with arbitrary support would do. However, with $(i)$ satisfied, the deeper implication $(i) \Rightarrow (ii)$ indicates that there must \emph{also} exist in $\scrPhat$ a reversible element \mbox{$(\delta_{\sigma} - \delta_{\sigma_0},\scrq)$}, with  the support of \scrq\ contained in $\Sigma_0$.
\end{rem}
\medskip

\begin{rem}
		For a Kelvin-Planck theory \theory\ we will say that states $\sigma$ and $\sigma'$ in $\Sigma$ are \emph{reversibly-connected}, denoted $\sigma \approx_{\,\scrR} \sigma'$, if $\sigma = \sigma'$ or if there is a $\scrq \in \MSigma$ such that \scrPhat\ contains both $(\delta_{\sigma} - \delta_{\sigma'}, \scrq)$ and its negative. Like the Carnot relation $\approx_{\,\scrC}$, the relation $\approx_{\,\scrR}$ is an equivalence relation in $\Sigma$. On any intersection of a $\approx_{\,\scrC}$-equivalence-class and an $\approx_{\,\scrR}$-equivalence-class there is essential uniqueness of Clausius-Duhem entropy-temperature pairs. Because in such an intersection all states are visited by both reversible connections and (reversible) Carnot elements, some readers might infer that these could only be equilibrium states.  See, however, Section \ref{sec:ConclRems,Part II}.
\end{rem}

\subsection[The Set of Entropy-Temperature Pairs and the Set of Processes]{A Relationship Between the Supply of Entropy-Temperature Pairs and the Supply of Processes}\label{sec:CDPairsDeterminePhat} For a Kelvin-Planck theory \theory\  in which there is an essentially unique Clausius-Duhem pair, Corollary \ref{cor:CDuniqueness&irreversibility} tells us that knowledge of a Clausius-Duhem pair determines \scrPhat\  completely, so long as there is at least one irreversible process. When for the theory there is not an essentially unique Clausius-Duhem pair, we can still ask about the relationship between the full \emph{set} of Clausius-Duhem pairs and \scrPhat. In particular, we can ask about circumstances under which \scrPhat\ coincides with the set of all members of \VSigma\ that comply with the Clausius-Duhem inequality for every choice of Clausius-Duhem pair --- that is, circumstances under which \scrPhat\  is identical to the set
\begin{equation} \label{eq:CDCompliantSet}
\scrQ := \{\  (v,w) \in \VSigma: \int_{\Sigma}\eta\, d\,v \geq \int_{\Sigma}\frac{d\,w}{
T}\ ,\forall (\eta, T) \in  CD\thinspace(\Sigma,\scrP)\},
\end{equation} 
where $CD\thinspace(\Sigma,\scrP)$ is the set of all Clausius-Duhem entropy-temperature pairs for \theory.

	From the positivity of Clausius-Duhem temperature scales it follows easily that the set
\begin{equation}
(0,\minus\scrM_+(\Sigma)) := \{(0,\nu) \in \VSigma: -\nu \in \scrM_+(\Sigma)\} 
\end{equation}
is contained in \scrQ. Thus, for \scrPhat\  to coincide with \scrQ\ it is \emph{necessary} that \scrPhat\  contain $(0,\minus\scrM_+(\Sigma))$. Less obvious is the fact that \emph{for \scrPhat\  to coincide with \scrQ\ it is both necessary and sufficient that \scrPhat\  contain $(0,\minus\scrM_+(\Sigma))$}. (This last assertion is largely a consequence of Lemma \ref{lemma:FirstBigLem}.) 

	Thus, if \scrPhat\  contains the simplest elements of \VSigma\  that comply with the Clausius-Duhem inequality for every entropy-temperature pair --- those elements of the form $(0,- \nu),\  \nu \in \scrM_+(\Sigma)$ --- then \scrPhat\  must contain \emph{all} elements of \VSigma\  that comply with the Clausius-Duhem inequality for every entropy-temperature pair. 
	
	Viewed as a process,  $\process := (0,- \nu),\  \nu \in \scrM_+(\Sigma)$ represents one that is cyclic and in which for every Borel set of states there is only heat \emph{emission}. The First Law then indicates that the work done on the body suffering the process, $\nu(\Sigma)$, is converted entirely into heat emitted to the body's exterior. It is not unreasonable to suppose that physical processes of this kind, or approximations to them, are naturally abundant.
	
\begin{rem}\label{rem:CDPairsDeterminePhat}
 When \scrPhat\ does not contain $(0,\minus\scrM_+(\Sigma))$, it is a consequence of Lemma \ref{lemma:FirstBigLem} that for any  member  $(v,w)$ of \scrQ\  that is not a member of \scrPhat\ there will nevertheless exist $\nu \in \scrM_+(\Sigma)$ such that $(v,w + \nu)$ is a member of \scrPhat.
\end{rem}

\section[Conjoined Thermodynamical Theories and Thermometers]{Conjoined Thermodynamical Theories and\\ Thermometers}\label{sec:Conjoined}

For the sake of simplicity and motivation, a thermodynamical theory \theory\ has been mostly viewed  as a description of processes that bodies composed of a particular material might experience. From this viewpoint, derived functions of state for a Kelvin-Planck theory, such as a specific-entropy function  $\eta: \Sigma \to \mathbb{R}$, were deemed to be attributes of the particular  material under consideration. In this interpretation of \theory, hotness levels in $\Sigma$ and their comparability relative to the ``hotter than" relation $\succ$ in \theory\ were regarded to be \emph{intrinsic} to the theory, ascertained \emph{only} by appeals to the set 
 of processes the material itself can or cannot experience. Indeed, we admitted the possibility that, for a particular Kelvin-Planck theory, \scrP\ might not be sufficiently adequate as to render every pair of hotness levels \emph{intrinsically} $\succ$-comparable or to make all Clausius-Duhem temperature scales essentially identical.  

In this section we will expand that interpretation of a thermodynamical theory to accommodate the idea that bodies composed of a particular material inhabit a world containing bodies made of still other materials, and that these various bodies can exchange heat. Indeed, two bodies in contact, composed of different materials, can be viewed as a single compound body that exchanges heat with its exterior. In such a case, heat might be absorbed from the exterior (of the compound body) by the first body, passed to the second body, and emitted to that exterior by the second body.\footnote{Recall that, in a thermodynamical theory, a heating measure for a process suffered by a body (including such a compound one) takes account only of heat exchange of the body with its exterior, not heat flows internal to the (compound) body.} 

With such  processes in mind, we will introduce the idea of the \emph{conjunction} of two thermodynamical theories --- that is, a broader thermodynamical theory that embraces the two theories and that, in addition, allows for processes of the type just described. We will be especially interested in situations in which one of theories characterizes a material that has special \emph{thermometric} properties relative to the other. 

We shall see that the presence of the thermometric theory in the conjunction can, in a certain sense, impart to the other theory a total ``hotter than" relation or an essentially unique Clausius-Duhem temperature scale where neither existed before.

\subsection{Conjoined Thermodynamical Theories}
\begin{definition}\label{def:Conjoined}
 Let $(\Sigma_1,\scrP_1)$ and $(\Sigma_2,\scrP_2)$ be thermodynamical theories having disjoint state spaces. We say that a thermodynamical theory $(\Sigma_3,\scrP_3)$  is a \textbf{conjunction of $(\Sigma_1,\scrP_1)$ and $(\Sigma_2,\scrP_2)$} if $\Sigma_3 = \Sigma_1\cup\Sigma_2$, if both $\scrP_1$ and $\scrP_2$ are essentially contained in $\scrP_3$, and if, for each $\process \in \scrP_3$, $\deltam(\Sigma_1) = 0$ and $\deltam(\Sigma_2) = 0$.\footnote{It is understood that $\Sigma_3$ has the disjoint union topology inherited from $\Sigma_1$ and $\Sigma_2$.}
\end{definition}

\begin{rem}\label{rem:EssentiallyContains}
This requires an explanation of what it means to say, for example, that  $\scrP_1$ is essentially contained in $\scrP_3$.  Note that if $(\Delta\scrm_1,\scrq_1)$  is an element of $\scrP_1$, then $\Delta\scrm_1$ and $\scrq_1$\  are both signed regular Borel measures on $\Sigma_1$. On the other hand,  if $(\Delta\scrm_3,\scrq_3)$  is an element of $\scrP_3$, then $\Delta\scrm_3$ and $\scrq_3$\  are both signed Borel measures on $\Sigma_3 = \Sigma_1\cup\Sigma_2$. In formal terms, then, a process in $\scrP_1$ cannot be a member of $\scrP_3$. Nevertheless, we say that $(\Delta\scrm_1,\scrq_1)$ is \emph{essentially contained in $\scrP_3$}\ if there is in $\scrP_3$\  a process $(\Delta\scrm_3,\scrq_3)$ such that $\Delta\scrm_3$ and $\scrq_3$\ take the value zero on every Borel set in $\Sigma_2$ and agree in value with $\Delta\scrm_1$ and $\scrq_1$\ on every Borel set of $\Sigma_1$. 
\end{rem}

\subsection{Thermometers for a Kelvin-Planck Theory} \label{subsec:ThermomDef} Throughout remainder of Section \ref{sec:Conjoined}, \theory\  is a generic Kelvin-Planck theory with hotness levels \scrH. In particular, we do not presume that the hotness levels of \scrH\  are totally ordered by $\succ$, the ``hotter than" relation in \theory. \emph{However,  $(\Sigma_{\Theta},\scrP_{\Theta})$ will hereafter designate a different Kelvin-Planck theory with hotness levels $\scrH_{\Theta}$, this time totally ordered, according to Definition \ref{def:HotterThan2}, by the \emph{hotter than} relation $\succ_{\Theta}$ in $(\Sigma_{\Theta},\scrP_{\Theta})$.} It will be understood that $\Sigma$ and $\Sigma_{\Theta}$ are disjoint.

\begin{definition}\label{def:thermom} The Kelvin-Planck theory $(\Sigma_{\Theta},\scrP_{\Theta})$\ is a \textbf{thermometer} for \theory\  if there is a Kelvin-Planck conjunction of $(\Sigma_{\Theta},\scrP_{\Theta})$\ and \theory, say $(\Sigma_C,\scrP_C)$, having the following property: For each $\sigma \in \Sigma$ there is a $\sigma_{\theta} \in \Sigma_{\Theta}$ such that both $(0,\delta_{\sigma} - \delta_{\sigma_{\Theta}})$ and its negative are members of $\hat{\scrP_C} := \mathrm{cl}\ [\Cone(\scrP_C)]$. In this case  $(\Sigma_C,\scrP_C$) is a \textbf{thermometric conjunction} of $(\Sigma_{\Theta},\scrP_{\Theta})$\ and \theory.
\end{definition}

\begin{rem}  The defining property amounts to a requirement that for each $\sigma \in \Sigma$ there is a $\sigma_{\theta} \in \Sigma_{\Theta}$ such that, in the conjunction, $\sigma$ and $\sigma_{\theta}$ are of the same hotness. See Appendix A, in particular Remark \ref{rem:ThermomHeatTransfer}, for a description of how the required passive heat transfers might come about in a natural way.
\end{rem}

\begin{rem} Here it will be helpful to regard  $(\Sigma_{\Theta},\scrP_{\Theta})$\ as a mathematical encoding of the thermodynamic properties of a thermometric material --- that is, a material which, for the purposes of measuring hotness, can finely probe, by means of heat transfer processes, a different target material, characterized by \theory; the thermometric material assigns to each state of the target material a hotness level in $\scrH_{\Theta}$.  In turn, each such hotness level is, relative to some chosen Clausius-Duhem temperature scale for $(\Sigma_{\Theta},\scrP_{\Theta})$, associated with a numerical  value of temperature.
\end{rem}

\begin{rem}[\emph{Conditions sufficient to ensure that $(\Sigma_{\Theta},\scrP_{\Theta})$ is a thermometer for \theory}] Suppose that $(\Sigma_C,\scrP_C)$ is a Kelvin-Planck conjunction of $(\Sigma_{\Theta},\scrP_{\Theta})$\ and \theory. In Proposition \ref{prop:SuffCondsThermom}  below we  assert  that  $(\Sigma_{\Theta},\scrP_{\Theta})$\ is a thermometer for \theory\  if the conjunction satisfies some very weak and natural requirements: that $\Sigma_{\Theta}$ is connected and that in the conjunction there is a kind of universal comparability of the states in $\Sigma$ and those in $\Sigma_{\Theta}$ with respect to the weakly-hotter-than relation $_w{\succ_C}$ in the conjunction.\footnote{In the spirit of Definition \ref{def:HotterThan1}, we say that state $\sigma'_C$ is weakly hotter than state $\sigma_C$ if the two states are of different hotnesses and\ $\widehat{\scrP_C}$ contains an element of the form $(0,\delta_{\sigma'_C} - \delta_{\sigma_C} + \nu)$, with $\nu \in \scrM_+(\Sigma_C)$. In particular, $\nu$ can be the zero measure, in which case there is a passive heat transfer from  $\sigma'_C$ to $\sigma_C$.}

\begin{proposition}\label{prop:SuffCondsThermom} Let  $(\Sigma_C,\scrP_C)$ be a Kelvin-Planck conjunction of the Kelvin-Planck theories $(\Sigma_{\Theta},\scrP_{\Theta})$\ and  \theory. Then  $(\Sigma_{\Theta},\scrP_{\Theta})$ is a thermometer for \theory\  if the following conditions are satisfied:
\begin{enumerate}[(i)]
\item{$\Sigma_{\Theta}$ is connected.}
\item{Any two states, one in  $\Sigma$ and the other in $\Sigma_{\Theta}$, that are not of the same hotness in  $(\Sigma_C,\scrP_C)$\  are $_w{\succ_C}$-comparable.}
\item{For each $\sigma \in \Sigma$  there is a state $\sigma_{\theta}\in \Sigma_{\Theta}$  such that $\sigma_{\theta}$ $_w{\succ_C}$ $\sigma$ and also\\ a state $\sigma'_{\theta}\in \Sigma_{\Theta}$  such that $\sigma$ $_w{\succ_C}$ $\sigma'_{\theta}$ .}
\end{enumerate}
\end{proposition}

\begin{rem}[\emph{Pervasiveness of $_w{\succ_C}$-comparability}] \label{rem:PervasiveWHotterThan} It is important to note that in order for two states of different hotness to be $_w{\succ_C}$-comparable it is enough that there be a passive heat transfer from one to the other. Appendix \ref{app:TransientPassiveHeatTransfer} suggests that such a transfer will take place whenever material samples in the two different states are brought into contact, however briefly.
\end{rem}

\begin{proof}[Proof of Proposition \ref{prop:SuffCondsThermom}] It must be shown that for each $\sigma \in \Sigma$ there is a $\sigma_{\theta} \in \Sigma_{\Theta}$ such that $\pm(0,\delta_{\sigma} - 
\delta_{\sigma_{\theta}})$ are members of 
$\widehat{\scrP_C}$. With $\scrT_C\, $ denoting the set of all Clausius-Duhem temperature scales for $(\Sigma_C,\scrP_C)$, this is equivalent by Theorem \ref{thm:EqualHotness} to showing that for each $\sigma \in \Sigma$ there is a $\sigma_{\theta} \in \Sigma_{\Theta}$ such that $T(\sigma) = T(\sigma_{\theta})$ for all $T \in \scrT_C$.

	Suppose on the contrary that there is a $\sigma^* \in \Sigma$ such that for each $\sigma_{\theta} \in \Sigma_{\Theta}$ there is a $\bar{T} \in \scrT_C$ such that $\bar{T}(\sigma^*) \neq \bar{T}(\sigma_{\theta})$. Let 
\begin{equation}
U_{\geq} := \{\sigma_{\theta} \in \Sigma_{\Theta}: T(\sigma_{\theta}) \geq T(\sigma^*), \forall T \in \scrT_C\} 
\end{equation}
and
\begin{equation}
U_{\leq} := \{\sigma_{\theta} \in \Sigma_{\Theta}: T(\sigma_{\theta}) \leq T(\sigma^*), \forall T \in \scrT_C\} .
\end{equation}
By supposition these sets are disjoint. From Theorem \ref{thm:HotterThan1AndTemp} and \emph{(iii)} both sets are non-empty. From the same theorem and \emph{(ii)}, the union of the two sets is $\Sigma_{\Theta}$. Because $U_{\geq}$ and $U_{\leq}$ are both closed and each is the (open) complement of the other, $\Sigma_{\Theta}$ is the union of two disjoint open sets, in violation of $(i)$.
\end{proof}

\end{rem}

	In preparation for the next section we posit the following definition:
	
\begin{definition}\label{def:IdealThermom} A thermometer $(\Sigma_{\Theta},\scrP_{\Theta})$  for a given Kelvin-Planck theory is an \textbf{ideal thermometer} for it if all Clausius-Duhem temperature scales for $(\Sigma_{\Theta},\scrP_{\Theta})$ are positive multiples of some fixed one.
\end{definition}

\begin{rem}[\emph{Approximate realization of ideal thermometers}]\label{rem:IdealThermomRealization} If $(\Sigma_{\Theta},\scrP_{\Theta})$ is an ideal thermometer for a Kelvin-Planck theory \theory\  then, in addition to its thermometric properties, $(\Sigma_{\Theta},\scrP_{\Theta})$ must satisfy all of the equivalent conditions stipulated in Theorem \ref{thm:UniquenessTemp}. In particular, $\scrPhat_{\Theta}$ must contain a rich supply of Carnot elements. This might be the case, for example, when $(\Sigma_{\Theta},\scrP_{\Theta})$ describes the thermodynamics of a gas such as nitrogen or helium, with $\Sigma$ consisting of pairs of the form $(p,v)$, with $p$ denoting the local pressure and $v$ denoting the local specific volume. In this case,  Carnot elements in $\scrPhat_{\Theta}$ might derive from  Carnot cycles specified by paths in $\Sigma_{\Theta}$, as depicted in standard text books.\footnote{For an ideal gas with processes as indicated in textbooks, the empirical ideal gas temperature scale, given by $T(p,v) := \frac{M}{R}{pv}$, has the properties of a Clausius-Duhem temperature scale. Here {M} is the molecular weight  of the gas and $R$ is the ideal gas constant. Under wide-ranging conditions, helium approximates an ideal gas very well.}

	Of course, the gas described by $(\Sigma_{\Theta},\scrP_{\Theta})$ must also satisfy the requirements of a thermometer for  \theory. The latter might, for example, describe bodies consisting of liquid mixtures in which chemical reactions occur among  a collection of several specified molecular species. In that case, most elements of $\Sigma$ would correspond to local mixture states in which chemical reaction equilibrium does not prevail. Nevertheless, Proposition \ref{prop:SuffCondsThermom}, Appendix \ref{app:TransientPassiveHeatTransfer}, and Remark \ref{rem:PervasiveWHotterThan} indicate how, for $(\Sigma_{\Theta},\scrP_{\Theta})$, the thermometric requirements of Definition \ref{def:thermom} might be satisfied by means of brief contacts between the gas and samples of the reacting liquid mixture. 
\end{rem}
	
\subsection[Properties Imparted by the Existence of a Thermometer]{Properties Imparted to a Kelvin-Planck Theory by the Existence of a Thermometer} 
\label{subsec:ThermomResults}
The following theorem describes a sense in which the existence of thermometer for a Kelvin-Planck  system \theory\ can impart to it properties that were not there intrinsically.

\begin{theorem}\label{thm:ThermomTheorem} Let \theory\  be a Kelvin-Planck theory in which the hotness levels in $\Sigma$ are not necessarily totally ordered by\  $\succ$, the \emph{hotter than} relation in \theory. Moreover, suppose that $(\Sigma_{\Theta},\scrP_{\Theta})$\ is a thermometer for \theory. If  $(\Sigma_C,\scrP_C)$ is a thermometric conjunction of  \theory\ and $(\Sigma_{\Theta},\scrP_{\Theta})$, then
\begin{enumerate}[(i)]

\item the hotness levels in $\Sigma_C$  are totally ordered by $\succ_C$, the \emph{hotter than} relation in $(\Sigma_C,\scrP_C)$. As a result, any two states of $\Sigma$ not in the same $C$-hotness level are, in the sense of Definition \ref{def:HotterThanStates},  $\succ_C$-comparable. 

\item If the thermometer is ideal, then all Clausius-Duhem temperature scales for the conjunction are positive multiples of some fixed one. In particular, the restrictions to  $\Sigma$ of all Clausius-Duhem temperature scales for the conjunction differ by at most a positive multiple. 
\end{enumerate}
\end{theorem}

\smallskip

 The theorem tells us that, for any two states  $\sigma, \sigma' \in \Sigma$ that are not of the same hotness in $(\Sigma_C,\scrP_C)$, we either have $\sigma' \succ_C  \sigma$ or $\sigma \succ_C  \sigma'$, this despite the fact that the same two states might not be intrinsically $\succ$-comparable in \theory. The enhanced comparability results from the presence of the thermometer in the larger conjoined theory $(\Sigma_C,\scrP_C)$, a presence that provides for more processes with which hotness comparisons can be made.

	Similarly, even when the Clausius-Duhem temperature scales for \theory\ are not all positive multiples of some fixed one (reflecting the absence of a sufficiently rich supply of Carnot elements in $\hat{\scrP}$), it will nevertheless be the case that, restricted to $\Sigma$, all Clausius-Duhem scales for the larger conjunction will be a positive multiple of some fixed one, so long as the thermometer is ideal --- that is, so long as the thermometer itself has an essentially unique Clausius-Duhem temperature scale. The essential uniqueness of Clausius-Duhem temperature scales for the conjunction derives from the richer supply of Carnot elements in $\hat{\scrP}_C$.  In Appendix \ref{app:NewCarnotElements} we describe a hypothetical physical scenario in which $\hat{\scrP}_C$ contains a Carnot element operating between two states of $\Sigma$ while $\hat{\scrP}$\ contains no such Carnot element.

\smallskip
\begin{proof}[Proof of Theorem \ref{thm:ThermomTheorem}]  We begin with some preliminary remarks: Because $(\Sigma_C,\scrP_C)$ is a Kelvin-Planck theory, there exists for it an entropy-temperature pair, both functions having domain $\Sigma_C$, that satisfies the Clausius-Duhem inequality for all processes in $\scrP_C$. Let $(\eta_{\,C},T_{C})$ be any such pair. Because $\scrP_{\Theta}$ is essentially contained in $\scrP_C$, it is apparent that the restrictions of $\eta_{\,C}$ and $T_C$ to $\Sigma_{\Theta}$ constitute a Clausius-Duhem pair for $(\Sigma_{\Theta},\scrP_{\Theta})$. In particular, the restriction of $T_C$ to $\Sigma_{\Theta}$ is a Clausius-Duhem temperature scale for $(\Sigma_{\Theta},\scrP_{\Theta})$. Therefore, whenever $h'_{\Theta}$ and $h_{\Theta}$ are hotness levels  in $\Sigma_{\Theta}$ such that $h'_{\Theta} \succ_{\theta} h_{\Theta}$ we must have $T_C(\sigma'_{\theta}) > T_C(\sigma_{\theta})$ for all $\sigma'_{\theta} \in h'_{\Theta}$  and   $\sigma_{\theta} \in h_{\Theta}$. Moreover, if $\sigma'_{\theta}$ and $\sigma_{\theta}$ are of the same hotness in $(\Sigma_{\Theta},\scrP_{\Theta})$, we must have $T_C(\sigma'_{\theta}) = T_C(\sigma_{\theta})$.

\smallskip

\noindent
\emph{Proof of (i)}.	We need to show that, if $h'_C$ and $h_C$ are distinct hotness levels for $(\Sigma_C,\scrP_C)$, then $h'_C$ and $h_C$ are $\succ_C$-comparable in the sense of Definition \ref{def:HotterThan2}. From properties of the thermometer,  every state in $\Sigma$ is of the same $\succ_C$-hotness as some state in $\Sigma_{\Theta}$. From this it follows that every hotness level for $(\Sigma_C,\scrP_C)$ contains a representative from $\Sigma_{\Theta}$. Suppose, then, that $\sigma'_{\theta}$ and $\sigma_{\theta}$ are such representatives taken from $h'_C$ and $h_C$,  respectively. Again from properties of $(\Sigma_{\Theta},\scrP_{\Theta})$, it must be the case that, relative to $(\Sigma_{\Theta},\scrP_{\Theta})$, the hotness levels $h'_{\Theta} \subset \Sigma_{\Theta}$ and $h_{\Theta}\subset \Sigma_{\Theta}$,  containing  $\sigma'_{\theta}$ and $\sigma_{\theta}$, are either $\succ_{\Theta}$-comparable or else they coincide.
		
		If  $\sigma'_{\theta}$ and $\sigma_{\theta}$ are of the same $\succ_{\Theta}$-hotness, then from the preliminary remarks above we have $T_C(\sigma'_{\theta}) = T_C(\sigma_{\theta})$ for every $T_C(\cdot)$ in the set of Clausius-Duhem temperature scale for $(\Sigma_C,\scrP_C)$. From Theorem \ref{thm:EqualHotness} it follows that $\sigma'_{\theta}$ and $\sigma_{\theta}$ are of the same hotness in $(\Sigma_C,\scrP_C)$. This, however, contradicts the supposition that $h'_C$ and $h_C$ are distinct.
		
		Suppose, then, that $h'_{\Theta}$ and $h_{\Theta}$ are $\succ_{\Theta}$-comparable, with $h'_{\Theta}\succ_{\Theta} h_{\Theta}$. From Theorem \ref{thm:HotterThan2AndTemp} and the preliminary remarks above, we have $T_C(\sigma'_{\theta}) > T_C(\sigma_{\theta})$ for each choice $T_C(\cdot)$ of Clausius-Duhem temperature scale for $(\Sigma_C,\scrP_C)$. From Definition \ref{def:HotterThanStates} and Corollary \ref{cor:HotterThanStates} if follows that $h'_C$  is $\succ_C$-comparable to  $h_C$, with  $h'_C\succ_C$ $h_C$.
\medskip

\noindent
\emph{Proof of (ii)}. Suppose that all Clausius-Duhem temperature scales for $(\Sigma_{\Theta},\scrP_{\Theta})$ are positive multiples of some fixed one. We want to show that the same is true of all Clausius-Duhem temperature scales for $(\Sigma_C,\scrP_C)$. Let $\bar{T}_C: \Sigma_C \to \RP$ and $T_C: \Sigma_C \to \RP$ be Clausius-Duhem temperature scales for $(\Sigma_C,\scrP_C)$. Moreover, let $\sigma^*_{\theta}$ be a fixed state in $\Sigma_{\Theta}$. It will be enough to show that
\begin{equation}\label{eq:TempsOnSigmaC}
\frac{\bar{T}_C(\sigma)}{T_C(\sigma)} = \frac{\bar{T}_C(\sigma^*_{\theta})}{T_C(\sigma^*_{\theta})},\quad \forall \sigma \in \Sigma_C.
\end{equation}
From properties of the thermometer $(\Sigma_{\Theta},\scrP_{\Theta})$, each $\sigma \in \Sigma_C$ is of the same $(\Sigma_C,\scrP_C)$-hotness as a state in $\Sigma_{\Theta}$, denoted here as $\sigma_{\theta}$. Because Clausius-Duhem temperature scales for  $(\Sigma_C,\scrP_C)$ assign the same value to all states in $\Sigma_C$ of the same $(\Sigma_C,\scrP_C)$\,-\,hotness, \eqref{eq:TempsOnSigmaC} is equivalent to 

\begin{equation}\label{eq:TempsOnSigmaC2}
\frac{\bar{T}_C(\sigma_{\theta})}{T_C(\sigma_{\theta})} = \frac{\bar{T}_C(\sigma^*_{\theta})}{T_C(\sigma^*_{\theta})},\quad \forall \sigma_{\theta} \in \Sigma_{\Theta}.
\end{equation}
From the preliminary remarks at the very beginning of the proof, the restriction to $\Sigma_{\Theta}$ of any Clausius-Duhem temperature scale for $(\Sigma_C,\scrP_C)$ is a Clausius-Duhem temperature scale for $(\Sigma_{\Theta},\scrP_{\Theta})$. That \eqref{eq:TempsOnSigmaC2} holds follows from the fact that all Clausius-Duhem temperature scales for $(\Sigma_{\Theta},\scrP_{\Theta})$ are identical up to a positive multiple.
\end{proof}

\subsection[Consistency of All Thermometers for a Kelvin-Planck Theory]{Ensured Consistency of All Thermometers for a Kelvin-Planck Theory}

 For the Kelvin-Planck theory  \theory,  we will suppose throughout this subsection that    $(\Sigma_{\Theta1}, \scrP_{\Theta1})$ and $(\Sigma_{\Theta2}, \scrP_{\Theta2})$ are two different thermometers (with $\Sigma_{\Theta1} \cap \Sigma_{\Theta2} = \emptyset$) and that $(\Sigma_{C1},\scrP_{C1})$ and $(\Sigma_{C2},\scrP_{C2})$ are, respectively, thermometric conjunctions of the two thermometers with \theory. 

	We want to show that, if the co-existence of the two (Kelvin-Planck) thermometric conjunctions does not, by virtue of that coexistence, conflict with the Kelvin-Planck Second Law, then the two thermometric conjunctions, each derived from a different thermometer, will impart to $\Sigma$ precisely the same hotter-than relations.  Moreover, if both thermometers are ideal,  then both conjunctions will impart the same (essentially unique) Clausius-Duhem temperature scale on $\Sigma$.

\begin{samepage}	
\begin{definition}	\label{def:KPCompatible}
	The thermometric conjunctions  $(\Sigma_{C1},\scrP_{C1})$ and $(\Sigma_{C2},\scrP_{C2})$ are \textbf{Kelvin-Planck compatible} if there is at least one Kelvin-Planck theory $(\Sigma_{C3},\scrP_{C3})$ in which $\Sigma_{C3} = \Sigma\cup \Sigma_{\Theta1}\cup   \Sigma_{\Theta2}$ and $\scrP_{C3}$ essentially contains $\scrP_{C1}$ and $\scrP_{C2}$ (in the sense of Remark \ref{rem:EssentiallyContains}). 
\end{definition}	
\end{samepage}
	
\begin{samepage}
\begin{theorem}[Consistency of Thermometers]\label{thm:ConsistencyOfTherm} Suppose that thermometric conjunctions $(\Sigma_{C1},\scrP_{C1})$ and $(\Sigma_{C2},\scrP_{C2})$  for \theory, corresponding to two different thermometers $(\Sigma_{\Theta1}, \scrP_{\Theta1})$ and $(\Sigma_{\Theta2}, \scrP_{\Theta2})$, are Kelvin-Planck compatible. 
\begin{enumerate}[(i)]
\item On $\Sigma$, the hotter-than relations derived from  $(\Sigma_{C1},\scrP_{C1})$ and $(\Sigma_{C2},\scrP_{C2})$  are identical. That is, if $\sigma'$ and $\sigma$ are states in $\Sigma$, then
\begin{equation}
 \sigma' \succ_{C1} \sigma \,\,\,\Leftrightarrow\,\,\,\sigma' \succ_{C2} \sigma.
\end{equation}
\item Suppose that for $j=1,2$ all Clausius-Duhem temperature scales for $(\Sigma_{Cj},\scrP_{Cj})$ are positive multiples of some fixed one, $T^*_{Cj}:\Sigma_{Cj} \to \RP$. Then, restricted to $\Sigma$, all Clausius-Duhem temperature scales for the two thermometric conjunctions are essentially identical. In particular, if  $\bar{T}^*_{Cj}:\Sigma \to \RP$ is the restriction of $T^*_{Cj}$ to $\Sigma$, then there is a positive number $\alpha$ such that $\bar{T}^*_{C2}(\cdot) = \alpha \bar{T}^*_{C1}(\cdot)$.
\end{enumerate}
\end{theorem}
\end{samepage}

\begin{proof} Throughout the proof, $(\Sigma_{C3},\scrP_{C3})$ is a fixed Kelvin-Planck theory satisfying the requirements of Definition \ref{def:KPCompatible}.

To prove $(i)$ we let $\sigma'$ and $\sigma$ be states in $\Sigma$ such that $\sigma' \succ_{C1} \sigma$. Corollary \ref{cor:HotterThanStates} then ensures that $T_{C1}(\sigma') > T_{C1}(\sigma)$ for every $T_{C1}$ that is a Clausius-Duhem temperature scale for $(\Sigma_{C1},\scrP_{C1})$. Contrary to what is to be proved, suppose that either $\sigma' \prec_{\,C2} \sigma$ or $\sigma' \sim_{C2} \sigma$. In these two cases,  we have, respectively, $T_{C2}(\sigma') < T_{C2}(\sigma)$ and $T_{C2}(\sigma') = T_{C2}(\sigma)$ for every $T_{C2}$ that is a Clausius-Duhem temperature scale for $(\Sigma_{C2},\scrP_{C2})$.
	
	The Kelvin-Planck theory $(\Sigma_{C3},\scrP_{C3})$ has at least one Clausius-Duhem temperature scale, say $T_{C3}$.  Because $\scrP_{C1}$ is essentially contained in $\scrP_{C3}$, it follows that the restriction of $T_{C3}$ to $\Sigma_{C1}= \Sigma \cup \Sigma_{\Theta1}$ is a Clausius-Duhem temperature scale for $(\Sigma_{C1},\scrP_{C1})$, in which case $T_{C3}(\sigma') > T_{C3}(\sigma)$.  Because $\scrP_{C2}$ is essentially contained in $\scrP_{C3}$, it also follows that the restriction of $T_{C3}$ to $\Sigma_{C2}= \Sigma \cup \Sigma_{\Theta2}$ is a Clausius-Duhem temperature scale for $(\Sigma_{C2},\scrP_{C2})$, in which case $T_{C3}(\sigma') < T_{C3}(\sigma)$ or  $T_{C3}(\sigma')= T_{C3}(\sigma)$. Thus, we have a contradiction. Proof that $\sigma' \succ_{C2} \sigma$ implies $\sigma' \succ_{C1} \sigma$ is similar.
	
	To prove $(ii)$ we again note, as in the proof of $(i)$, that for $j=1,2$ the restriction of $T_{C3}$ to $\Sigma_{Cj}= \Sigma \cup \Sigma_{\Theta_j}$ is a Clausius-Duhem temperature scale for $(\Sigma_{Cj},\scrP_{Cj})$. Given the hypothesis of (ii), then, $T^*_{Cj}:\Sigma_{Cj} \to \RP$ must, for $j=1,2$, be a positive multiple of the restriction of  $T_{C3}$ to $\Sigma_{Cj}$.  For this reason, $\bar{T}^*_{C2}(\cdot)$  must be a positive multiple of $\bar{T}^*_{C1}(\cdot)$.		
\end{proof}

\section[Equilibrium vs. Nonequilibrium Thermodynamics]{Concluding Remarks: Equilibrium vs. Non-equilibrium Thermodynamics} 
\label{sec:ConclRems,Part II}
In an attempt to  clarify
 and soften distinctions that are usually drawn between ``equilibrium" and ``nonequilibrium"  thermodynamics, we  review here what the theorems in this article and its precursor tell us about (the sometimes conflated) necessary and sufficient conditions for the very separate (also sometimes conflated) questions of existence and uniqueness of Clausius-Duhem entropy-temperature pairs.\footnote{Similar but less extensive observations were made in \cite{feinberg1983thermodynamics} and \cite{feinberg1986foundations}.}

The most important theorem of this two-part series is Theorem \ref{thm:ExistTempEnt}. It asserts that, for any thermodynamical theory consistent with the Kelvin-Planck Second Law, there exists a pair of continuous functions of state --- a specific entropy function and a thermodynamic temperature scale --- that, taken together, satisfy the Clausius-Duhem inequality for all processes the theory contains. This follows immediately from the Hahn-Banach Theorem. There is no requirement, either tacit or explicit, that the theory contain special processes, in particular reversible ones such as Carnot cycles or reversible processes that transform one state into another. Although brilliant classical textbook arguments do indeed show that a (presumed) abundance of  reversible processes is \emph{sufficient} to arrive at the \emph{existence} of a Clausius-Duhem pair, Theorem  \ref{thm:ExistTempEnt} tells us that reversible processes are \emph{not necessary} for that purpose.

However, \emph{existence} of these functions for a given Kelvin-Planck theory and their \emph{uniqueness} are very different matters. The larger the supply of processes, the smaller will be the set of Clausius-Duhem entropy-temperature pairs that comply with the Clausius-Duhem inequality for every process the theory contains.  Thus, if the set of entropy-temperature pairs for a given Kelvin-Planck theory is to be unique, either with respect to the temperature scale alone or with respect to both functions, then the set of processes extant in the theory must be sufficiently large as to ensure that the theory's set of entropy-temperature pairs is suitably narrow. 

Theorem \ref{thm:UniquenessTemp} indicates that, if a Kelvin-Planck theory is to have an essentially \emph{unique} temperature scale on its entire state space domain, it is \emph{necessary} that the theory   contain an abundance of (reversible) Carnot elements; in fact, there \emph{must} be a Carnot element operating between each pair of distinct states. If, in addition, the theory is to have a specific entropy function that is essentially unique on the entire state space, Theorem \ref{thm:CDPairUniqueness} \emph{requires}  that each pair of states also be connected by a reversible process. In each case, for \emph{uniqueness} on the entire state-space domain it is \emph{necessary} that \emph{every} state be``visited" by a \emph{reversible} process.

	In  the classical textbook picture, a reversible process has associated with it a path through a state space that can be traversed in both directions and in every detail. Such processes are usually regarded as ones that proceed so slowly that at each instant the body suffering the process can be regarded to be in a condition of equilibrium (or arbitrarily close to one). From this very classical perspective, an essentially unique Clausius-Duhem pair (or merely an essentially unique temperature scale) on the entire state space of a Kelvin-Planck theory would seem to require that all states in the theory be ``equilibrium" states.

	In this article, however, there is no notion of equilibrium.\footnote{Although the word \emph{equilibrium} is used often in thermodynamics textbooks, it is usually invoked intuitively and left without a precise definition, at least in a dynamical system sense.}  A reversible element of a Kelvin-Planck theory \theory\  is a mathematical object specified by Definition \ref{def:reversible}. It carries no requirement of a path through $\Sigma$ that is traversable in both directions, slowly or otherwise, nor is there any requirement of a two-way path through \MSigmaPl\ traversed by a body's   condition measure. 	Although the abundance  of reversible elements required by the uniqueness Theorems \ref{thm:UniquenessTemp} or \ref{thm:CDPairUniqueness} might indeed derive in one application or another from consideration of the idealized slow near-equilibrium processes depicted in textbooks,  \emph{that same abundance might derive from other sources and in other ways}.
		
	This is discussed in two appendices, described  in the following remarks.

\begin{rem}[\emph{Temperature scale uniqueness imparted to a Kelvin-Planck theory by the existence of an ideal thermometer}]	 Appendix \ref{app:NewCarnotElements} is meant as a companion to the more general \S s  \ref{subsec:ThermomDef} and \ref{subsec:ThermomResults}. In consideration of a Kelvin-Planck theory \theory\ that describes a hypothetical chemically reacting solution, we argue in Appendix \ref{app:NewCarnotElements} that, even when two  (nonequilbrium) states  $\sigma$ and $\sigma'$ in $\Sigma$ are unconnected by a (reversible) Carnot element $(0,c'\delta_{\sigma'} - c\delta_{\sigma})$ in \scrPhat, the existence\footnote{Recall Remark \ref{rem:IdealThermomRealization}.} of an ideal thermometer for \theory\  invariably gives rise to such a Carnot element in the conjunction of \theory\ with the thermometer. In that case,  Theorem \ref{thm:ThermomTheorem} ensures essential Clausius-Duhem temperature-scale uniqueness for the conjunction, in particular on all of $\Sigma$. 

	Note that \emph{it is in the conjunction that temperature-scale uniqueness on $\Sigma$ takes on its meaning and it is there that the presence of the Carnot element $(0,c'\delta_{\sigma'} - c\delta_{\sigma})$ is to be found}. It is in this broadened sense, involving the presumed availability of an ideal thermometer,  that Clausius-Duhem temperature-scale uniqueness becomes more universal in character than Clausius-Duhem entropy-function uniqueness, discussed in Section \ref{sec:EntropyUniqueness}. 
	
	Indeed, if members of a collection of distinct Kelvin-Planck theories, corresponding perhaps to a great variety of different materials, each had the same ideal thermometer, then for each of the pairwise thermometric conjunctions there would be  an essentially unique temperature scale, universally imposed across the collection by a single thermometer, regardless of whether state spaces of the individual Kelvin-Planck theories were restricted solely to states of equilibrium.
\end{rem}	

\begin{rem}[\emph{A reversible element realized in the limit by hypothetical physical processes that are very fast}] Theorems \ref{thm:CDPairUniqueness} and \ref{thm:EntropyUniqueness} indicate that, for a given Kelvin-Planck theory \theory\  with an essentially unique Clausius-Duhem temperature scale, essential uniqueness of a corresponding Clausius-Duhem specific-entropy function requires that every pair of states in $\Sigma$ be connected by a reversible element of \scrPhat. However, this   does not, by itself, require that all members of $\Sigma$ are, in some sense, equilibrium states. 

	Appendix \ref{App:ChemReactor} is intended to indicate that reversible elements in \scrPhat\  are not \emph{inextricably} linked to slow transitions along paths in \MSigmaPl \ consisting entirely of equilibrium conditions. There we suggest how  a reversible element of the form $(\delta_{\sigma'}-\delta_{\sigma}, \scrq) \in \scrPhat$ might arise in consideration of an idealized  chemical reactor, where neither  $\sigma$, $\sigma'$, nor the support of \scrq,  need be restricted to states of chemical equilibrium. Indeed, we indicate how such a reversible element might be the limit of a sequence of processes in \scrP\  corresponding to hypothetical physical realizations that occur at increasingly rapid rates.
\end{rem}
\bigskip
	
	In summary, then, if for a Kelvin-Planck theory \emph{existence} of Clausius-Duhem entropy and temperature functions of state are at issue, Theorem \ref{thm:EqualHotness} would seem to provide little support for those who might argue, perhaps based on standard textbook derivations, that the domains of those functions should be limited to equilibrium states. 
	
	If, however, essential \emph{uniqueness} of those same functions assumes critical importance in particular applications, then Theorems \ref{thm:UniquenessTemp} and \ref{thm:CDPairUniqueness} might, in some contexts, lend support to the claim that the domain of those functions should indeed be restricted to states of equilibrium. However, Appendices \ref{app:NewCarnotElements} and \ref{App:ChemReactor} should be kept in mind: reversible processes in the sense of those theorems might have a variety of physical origins, some involving nonequilibrium states.
	
	In any case we suspect that, in most applications, uniqueness of Clausius-Duhem entropy-temperature pairs will be considerably less consequential than their existence.

\appendix
\appendixpage
\addappheadtotoc

\section{An Example of Passive Unsteady Heat Transfer} 
\label{app:TransientPassiveHeatTransfer}

	For the purpose of motivation, we provided in Example \ref{ex:SimpleRodHeatTransfer} a hypothetical physical situation that, in a thermodynamical theory \theory, gave rise to a cyclic process in \scrPhat\  of the form $(0,\alpha(\delta_{\sigma'} - \delta_{\sigma}))$, with $\alpha > 0$, where $\sigma'$ and $\sigma$ are states in $\Sigma$. That the change of condition was the zero measure on $\Sigma$ (i.e., $\Delta \scrm = 0$) resulted from the fact that, in the example, the body suffering the process was in a temporally steady condition (as distinct from traditional thermodynamic equilibrium), so there was no change in the condition of the body between the process's inception and its termination. 

	Again for the purpose of motivation, it is our intent in this appendix to show, by means of a different hypothetical physical situation, that the same element in \scrPhat\ can derive from consideration of dynamic processes in which a steady condition is never present. The example is a simple toy model (e.g., one-dimensional, no motion), but it can be generalized to contain more complex and more natural features, suggesting that such processes will appear in \scrPhat\ whenever bodies come into momentary thermal contact.
	
	Consider, then, two samples of material, both samples described by the thermodynamical theory \theory, filling two slender tubes, each of length $L$ and small cross sectional area $A$, insulated along their extent, but not at their ends. The two samples are aligned along the $x$-axis, from $x = -L$ to $x = L$. The samples abut at $x=0$, separated by a perfectly heat-conducting barrier of negligible thickness.  A continuous function $r:[-L,L]\times [-t^*,t^*] \to \mathbb{R}$ describes the heat flux through tube cross-sections; that is, $r(x,t)$ is the rate of heat flow per unit cross-sectional area, in the positive $x$-direction,  through the cross-section at position $x$ and at time $t$. We will assume that $r(0,0)$ is positive. 
	
	We will also assume that there are two continuous functions, $\hat{\sigma}':[-L,0]\times [-t^*,t^*] \to \Sigma$ and $\hat{\sigma}:[0,L]\times [-t^*,t^*] \to \Sigma$ that give the point-wise state of the material  on the two sides of the barrier at each instant. We denote by $\sigma'$ and $\sigma$ material states, assumed to be different, contiguous to the two sides of the barrier at $t=0$.  That is,
\begin{equation}
\sigma' = \hat{\sigma}'(0,0)\quad  \mathrm{and} \quad \sigma = \hat{\sigma}(0,0).
\end{equation}

\begin{rem}\label{rem:ThermomHeatTransfer}
This picture is especially apt in our consideration of thermometric conjunctions in Section \ref{sec:Conjoined}, in which case \theory, the Kelvin-Planck theory considered here, would be replaced by the thermometric conjunction $(\Sigma_C,\scrP_C)$.  In such a context, $\sigma$ might represent a state of the thermometric material, while $\sigma'$ might represent a state of the material sample being probed.
\end{rem}
\medskip

	We will argue that, given the physical situation described, consideration of physical processes suffered by a sequence of sub-bodies along the tubes will give rise to a corresponding sequence of elements in $\scrPhat$\    that, for any $\alpha > 0$ having units of energy, converges in $\scrV(\Sigma)$ to  $(0,\alpha(\delta_{\sigma'} - \delta_{\sigma}))$. (It is assumed that physical processes suffered by all such sub-bodies are accounted for separately in \scrP.) For simplicity,\footnote{When these assumptions are dropped the outcome is essentially the same, but the analysis becomes more cumbersome.} we suppose that there is no motion and that the local material density on each side of the barrier is independent of spatial position and time, with $\rho'$ the density in the region $x \in [-L,0)$ and $\rho$ the density  for $x \in [0,L]$.
	
	For any $\xi$ and $\tau$, with $L > \xi > 0$ and $t^* > \tau > 0$,  we can calculate the process descriptor $\scrp(\xi,\tau) = (\Delta \scrm(\xi,\tau),\scrq(\xi,\tau)) \in \scrP$ that derives from consideration of the physical process suffered by the sub-body contained in the spatial interval $-\xi \leq x \leq \xi$  over the course of the time interval $[-\tau, \tau]$. 

	To specify a measure $\mu \in \scrM(\Sigma)$ it is enough to specify how $\mu$ integrates all continuous functions on $\Sigma$; that is, it is enough to specify the bounded linear functional  $\Gamma_{\mu}: C(\Sigma,\mathbb{R}) \to \mathbb{R}$ given by
\begin{equation}
\Gamma_{\mu}(f) = \int_{\Sigma} f\, d\mu,\quad  \forall f \in C(\Sigma,\mathbb{R}).
\end{equation}
The heating measure, $\scrq(\xi,\tau)$, for the process under consideration is given by 
\begin{align}\label{eq:q2ndAppendix}
&\Gamma_{\scrq(\xi,\tau)}(f) = \int_{\Sigma} f\, d\scrq(\xi,\tau) :=\\ &A\int_{-\tau}^{\tau}[f(\hat{\sigma}'(-\xi,t))\,r(-\xi,t) - f(\hat{\sigma}(\xi,t))\, r(\xi,t)]\ dt
,\ \   \forall f \in C(\Sigma,\mathbb{R}).\nonumber
\end{align}
The change of condition measure, $\Delta \scrm(\xi,\tau)$, is specified by the stipulation that, for all $g \in C(\Sigma,\mathbb{R})$,
\begin{align}\label{eq:DeltaM2ndAppendix}
&\Gamma_{\Delta \scrm(\xi,\tau)}(g) = \int_{\Sigma} g\, d\Delta \scrm(\xi,\tau) :=\\ &\rho'A\int_{-\xi}^{0}[g(\hat{\sigma}'(x,\tau)) - g(\hat{\sigma}'(x,-\tau))]\ dx +\rho A\int_{0}^{\xi}[g(\hat{\sigma}(x,\tau)) - g(\hat{\sigma}(x,-\tau))]\ dx \nonumber
\end{align}

Thus, if $\scrP \subset \scrV(\Sigma)$  contains descriptors of all physical processes our toy model admits, then, from consideration of the physical process corresponding to $\xi >0$ and $\tau>0$, we can conclude that $\scrP$ contains the process descriptor 
\begin{equation}\label{eq:ProcessXiTau}
\scrp(\xi,\tau) := (\Delta \scrm(\xi,\tau)),\scrq(\xi,\tau)),
\end{equation}
with $\Delta \scrm(\xi,\tau)$ and $\scrq(\xi,\tau)$ given by \eqref{eq:DeltaM2ndAppendix} and \eqref{eq:q2ndAppendix}. Therefore, if $\alpha$ is a positive constant (carrying units of energy) 
\begin{equation}\label{eq:ScaledProcessXiTau}
\alpha[2A\; r(0,0)\,\tau]^{-1}\scrp(\xi,\tau) := \alpha[2A\; r(0,0)\,\tau]^{-1}(\Delta \scrm(\xi,\tau),\scrq(\xi,\tau)),
\end{equation}	
is a member of $\Cone(\scrP)$. Our aim is to show that by judiciously taking a sequence of values of $\xi$ and $\tau$, shrinking to zero, \eqref{eq:ScaledProcessXiTau} will converge in $\scrV(\Sigma)$ to  \mbox{$(0,\alpha(\delta_{\sigma'} - \delta_{\sigma}))$}, which is to say that \mbox{$(0,\alpha(\delta_{\sigma'} - \delta_{\sigma}))$} is a member of $\scrPhat : = \mathrm{cl}\,[\Cone(\scrP)]$.

To show convergence in $\scrV(\Sigma)$ of \eqref{eq:ScaledProcessXiTau} to $(0,\alpha(\delta_{\sigma'} - \delta_{\sigma}))$ as  $\xi_n$ and $\tau_n$ approach $0$ in  at least certain selected ways,  we will argue that, for every choice of $f$ and $g$ in $C(\Sigma,\mathbb{R})$,
\begin{align}
\lim_{n \to \infty}\ \alpha [2A\; r(0,0)\,\tau_n]^{-1}(\int_{\Sigma} g\, d \Delta \scrm(\xi_n,\tau_n)  + \int_{\Sigma} f\, d\scrq(\xi_n,\tau_n)) =\\\int_{\Sigma} f\, d\,[\alpha (\delta_{\sigma'} - \delta_{\sigma})] = \alpha(f(\sigma') - f(\sigma)),\nonumber
\end{align}
provided that we take $\tau_n = \frac{1}{n}$ and $\xi_n = (\frac{1}{n})^2$. We first note from \eqref{eq:DeltaM2ndAppendix} that
\begin{align}
&\alpha\, [2A\; r(0,0)\,\tau_n]^{-1}|\int_{\Sigma} g\, d \Delta \scrm(\xi_n,\tau_n)|\\ &\leq \alpha\, [2A\; r(0,0)\,\tau_n]^{-1}[2A\;\xi_n\,(\ \rho' +\rho\;)\,|\,g\,|_{\,max}\,]\\
& = \alpha\, [\,r(0,0)\,]^{-1\,}[\,(\ \rho' +\rho\;)\,|\,g\,|_{\,max}\,|\,]\ \frac{\xi_n}{\tau_n}, \label{eq:SecondApp1stTerm}
\end{align}
where 
\begin{equation}
|\,g\,|_{\,max} := \max\{\, |\,g(\sigma)\,|\, :\, \sigma \in \Sigma\,\}
\end{equation} 
It is evident that, so long as we take $\tau_n = \frac{1}{n}$ and $\xi_n = (\frac{1}{n})^2$, the quantity shown in \eqref{eq:SecondApp1stTerm} will approach zero as $n \to \infty$.

	It remains to be argued that, with this same choice for $\tau_n$ and $\xi_n$,
\begin{equation} \label{eq:FinalGoal2ndApp}
\lim_{n \to \infty}\ \alpha\,[2A\; r(0,0)\,\tau_n]^{-1}( \int_{\Sigma} f\, d\scrq(\xi_n,\tau_n)) =\ \alpha(f(\sigma') - f(\sigma)).
\end{equation} 
From \eqref{eq:q2ndAppendix} it follows that
\begin{align}\label{eq:Finale2ndApp1}
&\alpha\,[2A\; r(0,0)\,\tau_n]^{-1}( \int_{\Sigma} f\, d\scrq(\xi_n,\tau_n)) =\\&\alpha\,[\,r(0,0)\,]^{-1}\{ \int_{-\tau_n}^{\tau_n}f(\hat{\sigma}'(-\xi_n,t))\,r(-\xi_n,t)\ \frac{dt}{2
\tau_n} - \int_{-\tau_n}^{\tau_n}f(\hat{\sigma}(\xi_n,t))\, r(\xi_n,t)\ \frac{dt}{2
\tau_n}\}.\nonumber
\end{align}

Therefore, to show that \eqref{eq:FinalGoal2ndApp} holds, with $\tau_n = \frac{1}{n}$ and $\xi_n = (\frac{1}{n})^2$,  it is enough to show that

\begin{equation}\label{eq:Finale2ndApp2a}
\lim_{n \to \infty}\ | \int_{-\tau_n}^{\tau_n}f(\hat{\sigma}'(-\xi_n,t))\,\frac{r(-\xi_n,t)}{r(0,0)}\ \frac{dt}{2
\tau_n} - f(\sigma')\,| = 0
\end{equation}
and
\begin{equation}\label{eq:Finale2ndApp2b}
\lim_{n \to \infty}\ | \int_{-\tau_n}^{\tau_n}f(\hat{\sigma}(\xi_n,t))\,\frac{r(\xi_n,t)}{r(0,0)}\ \frac{dt}{2
\tau_n} - f(\sigma)\,| = 0.
\end{equation}
However, these follow from continuity of the functions $r$, $f\circ\hat{\sigma}'$ and $f\circ\hat{\sigma}$.

\section[New Carnot Elements in a Thermometric Conjunction]{New Carnot Elements Arising in a Thermometric Conjunction}\label{app:NewCarnotElements}

	As a companion to \S\,\ref{subsec:ThermomResults}, we provide here a discussion, supplemented by a toy physical picture, to suggest how, for a Kelvin-Planck theory \theory\  endowed with an ideal thermometer $(\Sigma_{\Theta},\scrP_{\Theta})$,  their thermometric conjunction $(\Sigma_C,\scrP_C)$ can contain a Carnot element operating between two specified (perhaps non-equilibrium) states of $\Sigma$ even when \theory\  itself contains no such Carnot element. 

	For this purpose we suppose that \theory\ is a Kelvin-Planck theory of  liquid solutions composed of certain molecular species among which chemical reactions occur. We suppose also that $(\Sigma_{\Theta},\scrP_{\Theta})$ is an ideal thermometer for \theory, encoding the behavior of a thermometric material, which we will presume to be a perfect gas.

	As a preamble to the discussion, consider a single physical process involving heat transfer between two bodies --- one composed of the reacting liquid solution described by \theory\  and the other composed of the thermometric gas described by $(\Sigma_{\Theta},\scrP_{\Theta})$. In the theory \theory, the process will have associated with it a heating measure \scrq, defined on the 
Borel subsets of $\Sigma$. That same physical process, viewed from the perspective of the thermometric conjunction $(\Sigma_C,\scrP_C)$ will also have a heating measure $\scrq_C$ defined on the Borel subsets of $\Sigma_C = \Sigma \cup \Sigma_{\Theta}$. \emph{It should be clearly understood that the restriction of $\scrq_C$ to the Borel sets of $\Sigma$ can be very different from \scrq.}  This is because the heating measure in \theory\  captures details of heat exchange between the reacting solution and its exterior, \emph{an exterior that includes the gas thermometer}. For that same physical process, the corresponding heating measure in $(\Sigma_C,\scrP_C)$ captures the details of heat transfer between a composite body (the  solution sample taken with the thermometer) \emph{and the exterior of that composite body}. That is, in $(\Sigma_C,\scrP_C)$ the heating measure takes no account of heat transfer between the reacting liquid solution and the thermometric gas.
	
	Now let $\sigma \in \Sigma$ and $\sigma'  \in \Sigma$ be states of the reacting solution, not necessarily  states of chemical equilibrium. By properties of the thermometer, there are gas states $\sigma_{\theta}$ and $\sigma'_{\theta}$ in $\Sigma_{\Theta}$ (and therefore in $\Sigma_C$) such that, in the conjoined theory $(\Sigma_C,\scrP_C)$, $\sigma$ and $\sigma_{\theta}$  are of the same hotness, as are $\sigma'$ and $\sigma'_{\theta}$. Therefore, $\hat{\scrP}_C$\,  contains the passive heat transfers required by Definition \ref{def:thermom} between these liquid solution states and their corresponding gas states. 
	
	Because the ideal thermometer $(\Sigma_{\Theta},\scrP_{\Theta})$ has a unique Clausius-Duhem temperature scale, Theorem \ref{thm:UniquenessTemp} requires that  $\hat{\scrP}_{\Theta}$ contain a (reversible) Carnot element, say $(0,c\,\delta_{\sigma_{\theta}}-c'\,\delta_{\sigma'_{\theta}}) \in \scrV(\Sigma_{\Theta})$, operating between $\sigma_{\theta}$ and $\sigma'_{\theta}$. In physical terms, this Carnot element can be regarded as the limit of representations in $\scrPhat_{\Theta}$ of a sequence of classical ideal gas Carnot cycles (as usually depicted in pressure-volume space)  traversing two (decreasingly small) isothermal segments, one centered at $\sigma_{\theta}$ and the other at $\sigma'_{\theta}$.\footnote{In the sequence, the amount of gas experiencing each cycle needn't be the same.}
	
	Because $\scrP_{\Theta}$ is, in the sense of Remark \ref{rem:EssentiallyContains}, essentially contained in $\scrP_C$,   $(0,c\,\delta_{\sigma_{\theta}}-\,c'\,\delta_{\sigma'_{\theta}})$, viewed as a member of $\scrV(\Sigma_C)$, is a (reversible) Carnot element of $\hat{\scrP}_C$. As we indicated above, $\hat{\scrP}_C$ also contains (reversible) passive-heat-transfer elements of the form $(0,c'\,\delta_{\sigma'_{\theta}}-\,c'\,\delta_{\sigma'})$ and $(0,c\,\delta_{\sigma}-\,c\,\delta_{\sigma_{\theta}})$. Because $\hat{\scrP}_C$ is a convex cone the sum, \mbox{$(0,c\,\delta_{\sigma}-\,c'\,\delta_{\sigma'})$}, of these three members of $\hat{\scrP}_C$ having support entirely in $\Sigma$, is also a member of $\hat{\scrP}_C$.
	
	\emph{This is to say that in the thermometric conjunction $(\Sigma_C,\scrP_C)$ there is invariably a Carnot element operating between two (arbitrary) states  of the reacting liquid solution, $\sigma \in \Sigma$ and $\sigma'  \in \Sigma$, whether or not these be states of chemical equilibrium and even when the theory \theory \ of the reacting solution alone contains no such Carnot element.} 
	
	This is a consequence of the mathematics, deriving from the suppositions with which we began. To understand in more physical terms how such a Carnot element in $(\Sigma_C,\scrP_C)$ can emerge, even when absent in \theory, it will be useful to consider a toy  physical  picture meant to reflect the mathematics. The cartoon, like all cartoons, is imperfect, but it is only meant to be suggestive.  At the end of this appendix we will make two remarks about  how, in a much more extended exposition, certain of those imperfections might be mitigated. These remarks will draw on Appendix \ref{app:TransientPassiveHeatTransfer} and the appendix of this article's companion \cite{feinberg-lavineEntropy1}.
	 
	 In the cartoon, we imagine the Carnot  element  \mbox{$(0,c\,\delta_{\sigma}-\,c'\,\delta_{\sigma'})$} in $\hat{\scrP}_C$ to derive from a (limit) process of the following kind: A  solution sample in state $\sigma$, contiguous to the thermometric gas, rapidly absorbs a very small amount of heat, say $c$ calories, from an external bath while simultaneously passing that same small amount of heat to the thermometric  gas in state $\sigma_{\theta}$, all without appreciable changes to the solution sample. That heat is used  to drive a small isothermal segment of a Carnot cycle in the gas, that segment containing state $\sigma_{\theta}$. A small amount of heat, in the amount of $c'$ calories, is removed from the gas during the cycle's second small isothermal segment, that segment containing gas state $\sigma'_{\theta}$. The removed heat is rapidly passed to a different sample of the reacting solution, this one in state $\sigma'$, while an equal amount of heat is  simultaneously passed from there to a second external bath.
	
	Note that in this overall hypothetical physical process, \emph{viewed as one experienced by a physical conjunction of liquid solution and thermometric gas taken together}, the only heat exchange between the conjunction and \emph{the conjunction's exterior} is in the form of heat passage from the first external bath to the first solution sample (while in state $\sigma$) and then from the second solution sample (while in state $\sigma'$) to the second external bath. This is reflected in the process's codification as \mbox{$(0,c\,\delta_{\sigma}-\,c'\,\delta_{\sigma'})$} in $\hat{\scrP}_C$.
	
	However, viewed from the perspective of the reacting solution alone, described by the Kelvin-Planck theory \theory\  (as distinct from $(\Sigma_C,\scrP_C)$, the overall physical process indicated does not manifest itself as a  Carnot element. If it is kept in mind that the thermometer is part of the solution's exterior, as are the baths, it becomes apparent that there is no net absorption of heat from the solution's exterior by solution in either states $\sigma$ or $\sigma'$. This is to say that, in \theory, the heating measure for the overall physical process indicated is the zero measure in \MSigma.
	
\begin{rem}[\emph{Transient passive heat transfers between the baths and the solution samples}] In the cartoon, there is a transfer of a small amount of heat from the reacting-solution sample, while the sample is in a perhaps nonequilibrium state $\sigma$, to the thermometric material, while the thermometric material is in state $\sigma_{\theta}$. Because the reacting sample might be in a rapidly changing composition state, there arises the question of how the passive heat transfer $(0,c(\,\delta_{\sigma}-\,\delta_{\sigma_{\theta}})) \in \scrPhat_C$ could be realized. This was the general subject of Appendix \ref{app:TransientPassiveHeatTransfer}, with special reference to the thermometric setting in Remark \ref{rem:ThermomHeatTransfer}.  In rough terms, that element in $\scrPhat_C$ is  derived (in Appendix \ref{app:TransientPassiveHeatTransfer}) from consideration of very narrow material region straddling the sample-thermometer interface during a time interval of vanishingly small duration.

 Within the toy picture offered in this appendix, the initial reacting-liquid sample considered might be identified, in the sense of Appendix \ref{app:TransientPassiveHeatTransfer} (in particular Remark \ref{rem:ThermomHeatTransfer}), with a very thin sliver of liquid in the region $[-\varepsilon, 0]$ abutting the liquid-gas boundary, while the heat bath transmitting heat to that sample might be identified with the remaining liquid, residing in the region $[-L, -\varepsilon)$ exterior to the sliver.\footnote{In Remark \ref{rem:ThermomHeatTransfer}, $\sigma'$ would be identified with $\sigma$ here, while $\sigma$ there would be identified with $\sigma_{\theta}$ here.}
\end{rem}
\smallskip

\begin{rem}[\emph{About the addition of processes}] Prior to the introduction of the physical cartoon, the Carnot element $(0,c\,\delta_{\sigma}-\,c'\,\delta_{\sigma'})$ in $\scrPhat_C$\ derived mathematically as the sum of three other elements in $\scrPhat_C$, namely the passive heat transfers $(0,c'\,\delta_{\sigma'_{\theta}}-\,c'\,\delta_{\sigma'})$, $(0,c\,\delta_{\sigma}-\,c\,\delta_{\sigma_{\theta}})$, and the Carnot element in the thermometric gas, $(0,c\,\delta_{\sigma_{\theta}}-c'\,\delta_{\sigma'_{\theta}})$, viewed as a member of $\scrPhat_C$.

	That, for a natural thermodynamical theory, the closure of the cone of the process set should be closed under addition is a consequence of reasoning given in the appendix of \cite{feinberg-lavineEntropy1}. For the most part---but not entirely---this results from the supposition that two processes occurring in nature, suffered by different bodies, can be run in remote locations simultaneously to give a new natural process, suffered by the union of the two bodies, \emph{provided that the durations of the two separate processes are identical}. However, this was just one supposition in the appendix of \cite{feinberg-lavineEntropy1}. In light of still other natural suppositions, analysis in the appendix of \cite{feinberg-lavineEntropy1} indicates that, for the purpose of the additivity result, the simultaneity requirement is, in effect, inconsequential.
	
	This is mentioned here because, in the invocation of the physical cartoon, we have been casual about timing related to the two passive heat transfers between liquid and gas and also about timing related to the Carnot cycle in the gas. To be more precise, we have been casual about the timing of the physical processes (corresponding to members of $\scrP_C$) that \emph{approximate} those three limiting elements of $\scrPhat_C$.
	
	Discussion of such considerations would have made invocation of the cartoon significantly more complex than its didactic purpose warrants, but readers might want to keep in mind the appendix of \cite{feinberg-lavineEntropy1}.
\end{rem}

\section{ A Reversible Element in \scrPhat\  Involving Non\-equi\-lib\-rium States}
\label{App:ChemReactor}

It is the purpose of this appendix to indicate   how, in a theory \theory\  of  reacting mixtures,  there might arise in $\scrPhat\ := \mathrm{cl}\,[\Cone(\scrP)]$ a \emph{reversible} element of the form $(\delta_{\sigma'} - \delta_{\sigma}, \scrq)$, where neither $\sigma' \in \Sigma$ nor $\sigma \in \Sigma$ is a state of chemical equilibrium. 

We suppose that \theory\   describes  gaseous mixtures of $n$ molecular species $A_1, A_2,\dots, A_n$ that participate in a perhaps complex network of chemical reactions. The local states will be regarded to be elements of the form $(c, \theta) \in \R^{n+1}$, where $c := [c_1, c_2,\dots, c_n] \in \R^n$ is the vector of local molar concentrations of the $n$ species (moles per unit volume) and $\theta$ is the local temperature (perhaps on an empirical temperature scale). 

	To describe $\Sigma$, the full set of states for the theory, we first denote by $M := [M_1,M_2,\dots,M_n]$  the vector of molecular weights (mass per mole) of the species. For a fixed chosen positive value of $\rho^*$ (having units of mass per volume), the compact set
\begin{equation} \label{eq:OmegaDef}
\Omega : = \{ c \in \R^n : M\cdot c \leq \rho^*,\ c_i \geq 0,\  i = 1,2,\dots,n\} 
\end{equation}
is the set of all local molar concentration vectors consistent with a local mass density less than or equal to $\rho^*$. Hereafter we take $\Sigma = \Omega \times I$, where $I$ is a closed (temperature) interval of positive real numbers, perhaps very large. We suppose that, in the theory,  $\rho^*$ and $I$ are chosen to preclude from $\Sigma$ density and temperature extremes that are inappropriate to the model gaseous material under consideration.

For the mixture we presume that there are two smooth functions of state, \mbox{$\tilde{u}: \Sigma \to \R$} and $\tilde{f}:\Sigma \to \R^n$ with the following interpretations:  When $(c,\theta)$ is a local state in the mixture,  $\tilde{u}(c, \theta)$ is the local internal energy per unit volume, and $\tilde{f}(c, \theta) = [\tilde{f}_1(c,\theta),\dots,\tilde{f}_n(c,\theta)]$ is the vector of the local net molar production rates per unit volume of the $n$ species  due to the occurrence of all chemical reactions.
 
 Consider a mixture sample that fills a rigid closed vessel of constant volume, $V$, and suppose that the mixture remains spatially homogeneous at all times. That is,  at each instant the local state is the same everywhere. We presume that the local state is governed by the system of ordinary differential equations \eqref{eq:ReactorODEs}. 
 
 \begin{eqnarray}
\label{eq:ReactorODEs}
 \dot{c}_1 &=&  \tilde{f}_1(c,\theta)\nonumber\\
 \dot{c}_2 &=&  \tilde{f}_2(c,\theta)\nonumber\\
 &\vdots&\\
  \dot{c}_n &=&  \tilde{f}_n(c,\theta)\nonumber\\
 \frac{\partial \tilde{u}}{\partial \theta}(c,\theta)\,\dot{\theta} &=& - \nabla_c\, \tilde{u}(c,\theta) \cdot \tilde{f}(c,\theta) + Q(t)\nonumber
 \end{eqnarray}
 
 The overdot indicates differentiation with respect to time, and $Q(t)$ is the  rate  per unit volume at time $t$ of heat addition to the mixture within the reactor vessel. The first $n$ equations are molar balances of the species. The last equation reflects the First Law of Thermodynamics applied to the reactor under consideration: The rate of change of internal energy of mixture filling the rigid reactor vessel is equal to the rate at which heat is supplied to it.
 \medskip

\begin{rem}\label{rem:MixtureMassCons} Because the mass of the mixture filling the closed vessel is conserved and because the volume of the vessel is fixed, the density of the mixture remains constant in time, even as reactions cause the concentrations of the various species to change. If $\rho$ is the fixed density of mixture in the vessel, presumed less than $\rho^*$, then the evolving vector of molar concentrations, governed by \eqref{eq:ReactorODEs}, will forever remain in the set 
\begin{equation} \label{eq:ConservedMassSet}
\Gamma : = \{ c \in \R^n : M\cdot c = \rho,\ c_i \geq 0,\  i = 1,2,\dots,n\},  
\end{equation}
which is clearly contained in $\Omega$.
 
As a consequence,  $\tilde{f}$ must be such that \mbox{$M \cdot \tilde{f}(c,\theta)$} $= 0$ for all $c$ and $\theta$.
Under very weak assumptions  about the kinetics of the various reactions, the function $\tilde{f}$ also has the property that 
$\tilde{f}_i(c,\theta) \geq 0$ whenever $c_i = 0$, which is to say that the production rate of an absent species is not negative \cite{feinberg2019foundations}. Hereafter, we assume that $\tilde{f}$ has the property that any solution of \eqref{eq:ReactorODEs} that begins with $c$ initially in $\Gamma$ will be such that $c(t)$ remains in $\Gamma$, and therefore in $\Omega$, for all later times along the solution.
\end{rem}

\bigskip\smallskip

	By a solution of \eqref{eq:ReactorODEs} we will mean a set of $n+2$  functions of time, $c_1(\cdot)$, $c_2(\cdot)$,$\dots$, $c_n(\cdot)$, $\theta(\cdot)$, $Q(\cdot)$, that satisfy \eqref{eq:ReactorODEs} in some time interval (particular to that solution). So that we can focus specifically on the toy reactor described, we will confine our attention to solutions such that,
at the initial time, $[c(\cdot),\theta(\cdot)]$ takes values in $\Gamma \times I$. 
	
	Each such solution gives rise to a process $\process \in \VSigma$ in the following way: Let $[t_i,t_f]$ be the time interval of the solution, and let $\alpha:= V\rho$ be the mass of the mixture in the vessel. Then the change of condition induced by the solution is 
\begin{equation}
\deltam := \alpha(\delta_{[c(t_f),\,\theta(t_f)]} - \delta_{[c(t_i),\,\theta(t_i)]}).
\end{equation}
The  heating measure \scrq\  induced by the solution is defined by its action on continuous functions: For each continuous $\varphi\,: \Sigma \to \R$, 
\begin{equation}
\int_{\Sigma}\varphi\,\, d\scrq = \int_{t_i}^{t_f}\varphi\,\,(c(t),\theta(t))V\,Q(t)\,dt.
\end{equation}
Hereafter, we suppose that \scrP\  contains as a subset all processes  corresponding  to those solutions of \eqref{eq:ReactorODEs} that are consistent with the fixed mass $\alpha$ of the mixture  under consideration and initial conditions within $\Gamma \times I$.

	Let $c^{\,0}$ be a mixture composition in $\Gamma$ and let $\theta^{\,0}$ and $\theta^{\,^*}$ be temperatures such that $[c^{\,0},\,\theta^{\,0}]$ and $[c^{\,0},\theta^{\,^*}]$ are both in the interior of $\Sigma$. Moreover, for small $\varepsilon > 0$, let $\theta_{\varepsilon}(\cdot)$ be the temperature history defined by \begin{equation}\label{eq:TempProfile}
\theta_{\,\varepsilon}(t) := \theta^{\,0} + (\theta^* - \theta^{\,0})\frac{t}{\varepsilon},  \quad \forall t \in [0,\varepsilon].
\end{equation}
Consider the first $n$ equations of \eqref{eq:ReactorODEs}, with the temperature given by \eqref{eq:TempProfile} on the time interval $[0,\varepsilon]$. From Remark \ref{rem:MixtureMassCons} and the smoothness of $\tilde{f}$, the resulting $n$ equations admit a solution $c_{\,\varepsilon}(\cdot)$ on $[0,\varepsilon]$ satisfying the initial condition $c_{\,\varepsilon}(0) = c^{\,0}$. The full system \eqref{eq:ReactorODEs} of $n+1$ differential equations then admits the solution $c_{\,\varepsilon}(\cdot),\  \theta_{\,\varepsilon}(\cdot),\  Q_{\,\varepsilon}(\cdot)$, with $Q_{\,\varepsilon}(\cdot)$ calculated from $c_{\,\varepsilon}(\cdot)$, $\theta_{\,\varepsilon}(\cdot)$, and the last equation of \eqref{eq:ReactorODEs}.

	This solution  gives rise to the process $\scrp_{\varepsilon} = (\Delta \scrm_{\varepsilon}, \scrq_{\,\varepsilon})$, where
\begin{equation}
\deltam_{\varepsilon} := \alpha(\delta_{[c_{\varepsilon}(\varepsilon),\,\theta^*]} - \delta_{[c^{\,0},\,\theta^{\,0}]}),
\end{equation}
and $\scrq_{\varepsilon}$ is given by the requirement that, for every continuous $\varphi: \Sigma \to \R$,
\begin{equation}\label{eq:q_epsDef}
\int_{\Sigma}\varphi\, d\scrq_{\,\varepsilon} = \int_{0}^{\varepsilon}\varphi\ (c_{\,\varepsilon}(t),\theta_{\,\varepsilon}(t))V\,Q_{\,\varepsilon}(t)\,dt= \int_{0}^{1}\varphi\ (c_{\,\varepsilon}(\varepsilon s),\theta_{\,\varepsilon}(\varepsilon s))V\,Q_{\,\varepsilon}(\varepsilon s)\,\varepsilon ds.
\end{equation}	

	Because $\Sigma$ is compact, there is a number $A$ such that, for all $(c,\theta) \in \Sigma$, $\Vert\tilde{f}(c,\theta)\,\Vert \leq A$. From the first $n$ equations in \eqref{eq:ReactorODEs} it follows that $\Vert\, c_{\varepsilon}(t) - c^0\,\Vert \leq A\,\varepsilon$ for all $t \in [0, \varepsilon]$. As a result $\deltam_{\varepsilon}$ converges to
\begin{equation}\label{eq:deltam_0Def}
\deltam_{0} := \alpha(\,\delta_{\,[\,c^{0},\,\theta^*]} - \delta_{\,[\,c^{0},\,\,\theta^{\,0}\,]}\ )
\end{equation}
as $\varepsilon$ approaches $0$. Moreover, compactness of $\Sigma$ ensures that there is  number B such that, on $\Sigma$, $\vert\nabla_c\, \tilde{u} \cdot \tilde{f}\,\vert \leq B$. From the last equation in \eqref{eq:ReactorODEs} and \eqref{eq:TempProfile} it follows that, for all $s \in [0,1]$,
\begin{equation}\label{eq:QepsEstimate}
\vert\, \varepsilon\, Q_{\,\varepsilon}(\varepsilon s) - (\theta^* -\  \theta^{\,0}) \frac{\partial \tilde{u}}{\partial \theta}(c_{\,\varepsilon}(\varepsilon\,s)\,,\,\theta^{\,0} + (\theta^* - \theta^{\,0})\, s)\,\vert\ \leq\ \varepsilon B.
\end{equation}
Note that as $\varepsilon$ approaches $0$ the second term on the left of \eqref{eq:QepsEstimate} approaches
\begin{equation}
(\theta^* -\  \theta^{\,0}) \frac{\partial \tilde{u}}{\partial \theta}(c^{\,0}\,,\,\theta^{\,0} + (\theta^* - \theta^{\,0})\, s).
\end{equation}
From this it follows that, as $\varepsilon$ approaches $0$, the heating measure $\scrq_{\varepsilon}$ given by \eqref{eq:q_epsDef} converges to $\scrq_{0}$ defined by the requirement that, for each continuous $\varphi: \Sigma \to \R$,

\begin{equation}\label{eq:q_0Def}
\int_{\Sigma}\varphi\, d\scrq_{0} =  \int_{\theta_0}^{\theta^*}\varphi\,(c^0,\,\theta')\, \frac{\partial \tilde{u}}{\partial \theta}(c^{\,0}\,,\,\theta') V\,d\theta'.
\end{equation}	

	As $\varepsilon$ approaches $0$, then,  the family of processes $\scrp_{\varepsilon} = (\Delta \scrm_{\varepsilon}, \scrq_{\,\varepsilon})$ in \scrP\ converges to $\scrp_{\,0} =  (\Delta \scrm_{0}, \scrq_{0})$ in \scrPhat, with $\Delta \scrm_{0}$ and $\scrq_{0}$ given by \eqref{eq:deltam_0Def} and \eqref{eq:q_0Def}. To see that $-\scrp_{\,0}$ is also a member of \scrPhat\ it suffices to reverse the roles of $\theta^0$ and $\theta^*$. 
	
	Thus, we have in \scrPhat\  a \emph{reversible} element of the form
\begin{equation}\label{eq:RevReactorProc}
	(\delta_{[c^0,\ \theta^*]} - \delta_{[c^0,\,\theta^0]}\, ,\,\scrq_0\,).
\end{equation}
Note that $c^{\,0}$, $\theta^{\,*}$, and 	$\theta^{\,0}$ were chosen arbitrarily. \emph{Neither $[c^{\,0}, \theta^{\,*}]$ nor $[c^{\,0}, \theta^{\,0}]$	 need be a state of chemical equilibrium --- that is, a stationary solution of \eqref{eq:ReactorODEs} with $Q = 0$.}

\begin{rem} If \theory\ in our example is a Kelvin-Planck theory, and if $\bar{\eta}\,(\cdot)$ and $\eta\,(\cdot)$ are Clausius-Duhem specific-entropy functions corresponding to the same Clausius-Duhem temperature scale, then the Clausius-Duhem inequality and the presence of $(\Delta \scrm_{0}, \scrq_{0})$ in \scrPhat\ require that
\begin{equation}\label{eq:EntropyDiffEquality}
\bar{\eta}\,(c^0, \theta^*) - \bar{\eta}\,(c^0, \theta^0) =\eta\,(c^0, \theta^*) - \eta\ (c^0, \theta^0).
\end{equation}
Because $c^{\,0}$, $\theta^{\,0}$ and $\theta^*$ (and $\rho$) were chosen arbitrarily,  \eqref{eq:EntropyDiffEquality} indicates that for each $c \in \Omega$, the functions $\bar{\eta}\,(c, \cdot)$ and $\eta\,(c, \cdot)$ differ by at most a constant. In fact, if $T: \Sigma \to \RP$ is the Clausius-Duhem temperature scale to which the specific-entropy function $\eta(\cdot)$ corresponds, then for each $c \in \Omega$ the Clausius-Duhem inequality, \eqref{eq:deltam_0Def}, and \eqref{eq:q_0Def} require that $\eta(c,\cdot)$ be a function of the form

\begin{equation}
\eta\,(c,\theta) =  \frac{1}{M\cdot c}\int_{\theta_0}^{\theta}\frac{1}{T(c,\theta')}\, \frac{\partial \tilde{u}}{\partial \theta}(c,\,\theta') \,d\theta' + \gamma(c),
\end{equation}
where $\theta^{\,0}$ is some fixed value in $I$.
\end{rem}

\bigskip
\begin{rem} Reversible processes in the canonical picture are fictitious ones that proceed so slowly, and with such small changes, that they could never be completed. Nevertheless, they are regarded as processes that, in principle, can be approximated by real ones sufficiently well that a complete theory should embrace them in the limit.

	In the context of the reacting-mixture theory  considered  in this appendix, the limiting reversible process (corresponding to $\varepsilon = 0$) is approximated by processes of a very different kind: As $\varepsilon$ approaches zero, they complete increasingly quickly, with increasingly rapid changes in temperature, and with increasingly higher rates of heat transfer  (all sustained over  vanishingly small time intervals).
\end{rem}
\bigskip

\begin{rem}	
	In the hypothetical $\varepsilon$-parameterized processes described, temperature is presumed to be spatially uniform despite very rapid rates of heat transfer. In the case of conductive heat transfer to the mixture at the reactor wall, large values of heat \emph{flux} are associated with large values of spatial temperature gradients in the mixture at the mixture boundary.  However, even in the case of conductive heat transfer from the  exterior at the mixture boundary, a large value of $Q_{\varepsilon}(t)$ (rate of heat receipt per unit reactor volume) does not necessitate a large heat flux (rate of heat receipt per unit area) at the reactor walls.   For each $\varepsilon > 0$ we can imagine the reactor vessel to be a tall narrow circular cylinder of fixed radius $R_{\varepsilon}$, in which case the  heat flux at the cylinder wall would be $Q_{\epsilon}(t)R_{\varepsilon}/2$. By choosing  $R_{\varepsilon}$ sufficiently small, the instantaneous heat fluxes at the wall  (and presumably the temperature gradients there) can be kept as small as we wish.
\end{rem}

\section{Clausius vs$.$ Clausius-Duhem Temperature Scales} \label{app:ClausiusVsCD}

In a 1983 article \cite{feinberg1983thermodynamics} we examined the existence and properties of \emph{Clausius temperature scales} (as distinct from Clausius-Duhem temperature scales) for thermodynamic theories that respect the Kelvin-Planck Second Law. In that context, for a theory with state space $\Sigma$, a Clausius temperature scale $T: \Sigma \to \mathbb{R}_+$ was a continuous function that satisfies, for all \emph{cyclic} processes, the \emph{Clausius inequality}. The Clausius inequality is just the form that the Clausius-Duhem inequality takes for cyclic processes.  Thus, $T(\cdot)$ is a Clausius temperature scale if it is continuous and satisfies the Clausius inequality condition 
\begin{equation}\label{eq:ClInequality}
0 \geq \int_{\Sigma}\frac{d\scrq}{T},\quad \forall \scrq \in \scrC,
\end{equation}
where $\scrC \subset \MSigma$  is the set of heating measures associated with the theory's cyclic processes. In \cite{feinberg1983thermodynamics}, \scrC\ was called the set of \emph{cyclic heating measures}.

	Because the focus of \cite{feinberg1983thermodynamics} was entirely on cyclic processes, there was nothing in \cite{feinberg1983thermodynamics} corresponding to \scrP, the full set of processes central to this article (apart from an anticipatory description of \scrP\  in the concluding remarks of \cite{feinberg1983thermodynamics}). Instead, \scrC\ was taken in \cite{feinberg1983thermodynamics} as a primitive notion, part of the description of a \emph{cyclic heating system}, $(\Sigma,\scrC)$.\footnote{There the Kelvin-Planck Second Law took the form $\textrm{cl}\,(\Cone(\scrC))\; \cap\; \MSigmaPl = \{0\}$.}
	
	In light of the interpretation given to \scrC\  in \cite{feinberg1983thermodynamics}, we hereafter define  \scrC\  in this appendix as follows: For a Kelvin-Planck theory \theory,
\begin{equation}\label{eq:DefScrC}
\scrC := \{\scrq \in \MSigma: (0,\scrq) \in \scrP\}.
\end{equation}
By a \emph{Clausius temperature scale for a Kelvin-Planck theory \theory}, we mean a continuous function $T: \Sigma \to \mathbb{R}_+$ that satisfies  the condition \eqref{eq:ClInequality}, with \scrC\  as in \eqref{eq:DefScrC}. For \theory\  we denote by $\scrT_{Clausius}$ the set of all its Clausius temperature scales.

	Our interest is in the relationship between the set of Clausius temperature scales for a Kelvin-Planck theory \theory\ and its set $\scrT_{CD}$ of Clausius-Duhem temperature scales, described in Definition \ref{def:CDpair}.  Recall that a continuous function $T: \Sigma \to \mathbb{R}_+$ is a Clausius-Duhem temperature scale for \theory\ if there is a continuous (specific entropy) function $\eta: \Sigma \to \mathbb{R}$ such that
\begin{equation}\label{eq:CDIneq}
\int_{\Sigma}\eta\: d(\deltam)\  \geq\  \int_{\Sigma}\frac{d\scrq}{T}, \quad \forall\  (\deltam,\scrq) \in \scrP.
\end{equation}
Because \eqref{eq:CDIneq} is a more demanding requirement than is \eqref{eq:ClInequality} we always have
\begin{equation}
\scrT_{CD} \subset \scrT_{Clausius}.
\end{equation}
The following example demonstrates that	$\scrT_{Clausius}$ can in fact be larger than $\scrT_{CD}$ even when \scrP\ (as distinct from \scrPhat) is a closed convex cone.

\begin{example}\label{ex:ClScaleEx1} Here we consider a Kelvin-Planck theory \theory\  with state space $\Sigma$ consisting of just two states, labeled $1$ and $2$; $\Sigma$ is given the discrete topology. The process set \scrP\  is the closed convex cone consisting of all $\process \in \VSigma$ such that, with $\alpha$ taking all real values,
\begin{equation}\label{eq:scrPforEx1}
\Delta\scrm = \lambda\,\xi\,(\delta_2 - \delta_1),\quad \scrq = \lambda\,[(\alpha-1)\,\delta_1 + (\alpha+1)\,\delta_2], \quad \xi \geq \alpha^2,\  \lambda \geq 0.
\end{equation}

We will consider first the nature of Clausius-Duhem entropy-temperature pairs for the Kelvin-Planck theory \theory. In this case, a Clausius-Duhem entropy function $\eta: \{1,2\} \to \mathbb{R}$ amounts to a specification of two numbers $\eta_{\,1}$ and  $\eta_{\,2}$.  A Clausius-Duhem temperature scale $T: \{1,2\} \to \mathbb{R}_+$ amounts to a specification of two positive numbers $T_1$ and  $T_2$. For $(\eta,T)$ to constitute a Clausius-Duhem pair, it must satisfy the Clausius-Duhem inequality for all $\process \in \scrP$. If we let $\beta_1 = 1/T_1$ and $\beta_2 = 1/T_2$, this amounts to the requirement that
\begin{equation}
\xi\,(\eta_{\,2} - \eta_{\,1}) \geq (\alpha-1)\beta_1 + (\alpha +1)\beta_2,\quad \forall \alpha,\  \forall\; \xi \geq \alpha^2.
\end{equation}
From this it is apparent that we must have $\eta_{\,2}- \eta_{\,1} > 0$ and
\begin{equation}\label{eq:CDalpha1}
\alpha^2\,(\eta_{\,2} - \eta_{\,1}) -(\beta_1 + \beta_2)\alpha + (\beta_1-\beta_2) \geq 0,\quad \forall \alpha.
\end{equation}
It is not difficult to see that \eqref{eq:CDalpha1} can be satisfied only if
\begin{equation}\label{eq:beta1>beta2}
\beta_1 > \beta_2 \quad \mathrm{or,\, equivalently,}\quad T_2 > T_1.
\end{equation}
In fact, so long as \eqref{eq:beta1>beta2} is satisfied, the Clausius-Duhem inequality will be satisfied for all members of \scrP\ with $\eta$ chosen such as to have
\begin{equation}\label{eq:EtaCond}
\eta_{\,2} - \eta_{\,1} \geq \frac{(\beta_1 + \beta_2)^2}{4\,(\beta_1-\beta_2)}.
\end{equation}
Thus, for the Kelvin-Planck theory under consideration the set of Clausius-Duhem temperature scales (Definition \ref{def:CDpair}) is given by 
\begin{equation}
\scrT_{CD} := \{T: \{1,2\} \to \mathbb{R}_+: T_2 > T_1\}.
\end{equation}

We turn next to consideration of the set of \emph{Clausius temperature scales} for the same Kelvin-Planck theory \theory. In this case, the set of cyclic heating measures (corresponding to $\xi =0,\  \alpha = 0$) is given by\footnote{Note that in this case $\scrC := \textrm{cl}\,(\Cone(\scrC))$.}
\begin{equation}
\scrC := \{\scrq \in \MSigma: \scrq := \lambda\,(\delta_2 - \delta_1), \ \lambda\  \geq\ 0\}.
\end{equation}
Thus, a temperature function satisfies the Clausius requirement \eqref{eq:ClInequality} precisely when
\begin{equation}
0 \geq \frac{1}{T_2} - \frac{1}{T_1}.
\end{equation}
This is to say that the set of Clausius temperature scales is given by
\begin{equation}\label{eq:ClScaleforEx1}
\scrT_{Clausius} := \{T: \{1,2\} \to \mathbb{R}_+: T_2 \geq T_1\}.
\end{equation}
Note that $\scrT_{CD}$ is contained in $\scrT_{Clausius}$ but is not identical to it.
\end{example}
\medskip

In Example \ref{ex:ClScaleEx1}, \scrP\  was a closed convex cone. In the following example, which is highly similar to the preceding one, \scrP\ is not closed, and the distinction between $\scrT_{CD}$ and $\scrT_{Clausius}$ becomes substantially more pronounced.

\begin{example}\label{ex:ClScaleEx2} In this example we consider a Kelvin-Planck theory \theory\ that is identical to the one in Example \ref{ex:ClScaleEx1} apart from one difference: Whereas in Example \ref{ex:ClScaleEx1} the parameter $\alpha$ was permitted to take on all real values, here we restrict $\alpha$ to the \emph{nonzero} real values. Despite the difference, the set of Clausius-Duhem pairs here remains what it was in Example \ref{ex:ClScaleEx1}. In particular, we again have
\begin{equation}
\scrT_{CD} := \{T: \{1,2\} \to \mathbb{R}_+: T_2 > T_1\}.
\end{equation}

In this case, however, $\alpha \neq 0$, so \scrP\  contains no cyclic processes at all, apart from the trivial one in which $\process = (0,0)$. Thus, we have $\scrC = \{0\}$ and, as a result,
\begin{equation}
\scrT_{Clausius} := \{T: \{1,2\} \to \mathbb{R}_+: T_1 > 0,  T_2 > 0 \}.
\end{equation}
\end{example}
\begin{center}
***
\end{center}
\medskip

	Note that in the definition of \scrC, given in this appendix by \eqref{eq:DefScrC}, the cyclic heating measures derive from the cyclic processes contained only within the true process set \scrP, as distinct from the sometimes larger set $\scrPhat := \textrm{cl}\,(\Cone(\scrC))$. This definition of \scrC\ was motivated entirely by the physical interpretation  given to \scrC\ in \cite{feinberg1983thermodynamics}, where \scrC\  was merely described as the set of heating measures associated with cyclic processes. There was no mention or even a description of the fuller set of all processes (except in the concluding remarks).
	
	In the main body of this article a \emph{cyclic element} of a thermodynamic theory \theory\  is a member \process\ of \scrPhat\  such that $\Delta\scrm =0$. The set of  cyclic elements of \theory\ contains not only all the cyclic processes in \scrP\  but also members of \VSigma\ that are approximated arbitrarily closely by ``almost cyclic" processes (or positive multiples of them). Leaving the interpretation of \scrC\  in \cite{feinberg1983thermodynamics} aside, we could just as well have defined a Clausius temperature scale for a Kelvin-Planck theory to be a continuous function  $T: \Sigma \to \mathbb{R}_+$ such that	
\begin{equation}\label{eq:ClInequalityStrong}
0 \geq \int_{\Sigma}\frac{d\scrq}{T},\quad \forall \scrq \in \scrC^*,
\end{equation}
where $\scrC^{*} \in \MSigma$  is the set of heating measures associated with the theory's cyclic elements. More precisely,\footnote{Were \scrC\ in \cite{feinberg1983thermodynamics} identified with $\scrC^*$ as defined here, all mathematics would remain the same; only the interpretation would be different.}
\begin{equation}\label{eq:DefScrCStar}
\scrC^* := \{\scrq \in \MSigma: (0,\scrq) \in \scrPhat\}.
\end{equation}
For \theory\ we denote the set of Clausius temperature scales defined in this way by $\scrT^*_{Clausius}$. Because the condition \eqref{eq:ClInequalityStrong} is more demanding than \eqref{eq:ClInequality} we call members of $\scrT^*_{Clausius}$ the \emph{strong Clausius temperature scales} for \theory.

	We will now reconsider Examples \ref{ex:ClScaleEx1} and \ref{ex:ClScaleEx2} in light of these ideas. Because in Example \ref{ex:ClScaleEx1} \scrPhat\ is identical to \scrP, we have $\scrC^* = \scrC$, so there is no distinction between $\scrT^*_{Clausius}$ and $\scrT_{Clausius}$. Thus, for Example \ref{ex:ClScaleEx1} we have 	
\begin{equation}
\scrT^*_{Clausius}=\scrT_{Clausius} := \{T: \{1,2\} \to \mathbb{R}_+: T_2 \geq T_1\}.
\end{equation}
Every Clausius temperature scale is also a strong Clausius temperature scale. 
	
	In the case of Example \ref{ex:ClScaleEx2} we noted that there are no cyclic processes in \scrP\  apart from the trivial one, so $\scrC = \{0\}$, and 
\begin{equation}
\scrT_{Clausius} := \{T: \{1,2\} \to \mathbb{R}_+: T_1 > 0,  T_2 > 0 \}.
\end{equation}	
	 However, for the process set \scrP\ described in Example \ref{ex:ClScaleEx2},\   $\scrPhat := \textrm{cl}\,(\Cone(\scrP))$ is identical to the process set \scrP\  of Example \ref{ex:ClScaleEx1}, given by \eqref{eq:scrPforEx1} with $\alpha$ taking all real values. Thus, for Example \ref{ex:ClScaleEx2}, $\scrC^*$ is identical to \scrC\  in Example \ref{ex:ClScaleEx1}, and
\begin{equation}
\scrT^*_{Clausius} := \{T: \{1,2\} \to \mathbb{R}_+: T_2 \geq T_1\}.
\end{equation} 
Example \ref{ex:ClScaleEx2} indicates that for a Kelvin-Planck theory the set of strong Clausius temperature scales can be very different from the set of Clausius temperature scales. 

	In both examples the sets of Clausius scales and strong Clausius scales are different from the set of Clausius-Duhem temperature scales, which in both examples is given by 
\begin{equation}
\scrT_{CD} := \{T: \{1,2\} \to \mathbb{R}_+: T_2 > T_1\}.
\end{equation}
Note that in Example \ref{ex:ClScaleEx2} the set of Clausius-Duhem scales is very different from the set of Clausius temperature scales, but in both examples the set of \emph{strong} Clausius temperature scales resembles very closely the set of Clausius-Duhem temperature scales.	

	In some ways this last observation is surprising, for the Clausius-Duhem temperature scale requirement (Definition \ref{def:CDpair}) must take cognizance of the entire process set \scrP, not just the cyclic ones, while the strong Clausius temperature scale requirement \eqref{eq:ClInequalityStrong} takes cognizance only of the heating measures for Kelvin-Planck theory's cyclic elements. In fact, though, Theorem \ref{thm:DensityThm}\footnote{In the theorem statement it is understood that $C(\Sigma,\mathbb{R})$ is given the sup norm topology.} below indicates that the phenomenon exhibited by the examples is general.

\begin{theorem}\label{thm:DensityThm} For any Kelvin-Planck theory the set of Clausius-Duhem temperature scales is dense in the set of strong Clausius temperature scales.
\end{theorem}

	In the proof Theorem of \ref{thm:DensityThm}, a Hahn-Banach separation theorem (but a different version \cite{robertson1980topological}) will again play a central role: \emph{Let $V$ be a real locally convex topological vector space. If $A$ and $B$ are disjoint nonempty convex subsets of $V$ and $A$ is open, then there is a continuous linear function \mbox{$f:V \to \mathbb{R}$} and a constant $\alpha$ such that $f(x) < \alpha$ for all $x \in A$ and $f(x) \geq \alpha$ for all $x \in B$.} 
	
	Note that if $B$ is closed under positive multiplication --- that is, if $\lambda b$ is a member of $B$ for every $b \in B$ and every $\lambda > 0$ --- then  $\alpha$ can be taken to be zero. In the proof below the set labeled $CD$, which will play the role of $B$, is closed under positive multiplication.

\begin{proof}[Proof of Theorem \ref{thm:DensityThm}] 
	We consider a Kelvin-Planck theory \theory. Let $K$ be the linear subspace of $C(\Sigma,\mathbb{R})$ consisting of all constant functions. The equivalence relation $\sim$ in $C(\Sigma,\mathbb{R})$ defined by $f \sim g$ if and only if $f-g \in K$  gives rise to the quotient vector space $C_0(\Sigma) := C(\Sigma,\mathbb{R})/K$, with vectors consisting of the equivalence classes and vector space operations inherited from $ C(\Sigma,\mathbb{R})$ in the usual way.  We give  $C_0(\Sigma)$ the usual quotient topology, in which case $C_0(\Sigma)$\  and \MSigmaZ\  are mutually dual spaces.
	
	The equivalence class in $C_0(\Sigma)$ containing $f \in C(\Sigma,\mathbb{R})$  is denoted $[f]$. Note that for every measure 
$\mu \in \MSigmaZ$ and every $g\in [f]$ we have
\begin{equation}\label{eq:Intf}
\int_{\Sigma}g\,d\mu = \int_{\Sigma}f\,d\mu,
\end{equation}
so there is no ambiguity in the definition
\begin{equation}\label{eq:Intf2}
\int_{\Sigma}[f]\,d\mu := \int_{\Sigma}f\,d\mu.
\end{equation}

	Notwithstanding a slight abuse of language and identification of $\beta$ with $1/T$, we will say that $[\eta ]\in C_0(\Sigma)$ and  $\beta \in C(\Sigma,\mathbb{R}_+)$ constitute a Clausius-Duhem pair $([\eta],\beta)$ for \theory\ if 
\begin{equation}\label{eq:CDIneq[]}
\int_{\Sigma}[\eta]\: d(\deltam)\  \geq\  \int_{\Sigma}\beta\,d\scrq, \quad \forall\  (\deltam,\scrq) \in \scrPhat.
\end{equation}
In this way we can identify the set of Clausius-Duhem pairs for \theory\ with a subset $CD$ of the locally convex topological vector space
\begin{equation}
 \scrV^*(\Sigma) := C_0(\Sigma) \oplus C(\Sigma,\mathbb{R}). 
\end{equation}
What we have called \VSigma\   and $\scrV^*(\Sigma)$  are mutually dual.

Let $\beta_0$ be the reciprocal of a strong Clausius temperature scale for the Kelvin-Planck theory \theory, and let  $N$ be an open convex neighborhood of $\beta_0$ in $C(\Sigma,\mathbb{R}_+)$. We will show that $N$ contains a $\beta$ such that, for some $[\eta] \in C_0(\Sigma)$, the pair $([\eta],\beta)$ is a member of $CD$. Suppose on the contrary that the open convex set $C_0(\Sigma)\oplus N$ is disjoint from the convex  set $CD$, which is invariant under positive multiplication. Then, from the Hahn-Banach separation theorem stated just above, there is a vector $(\Delta \scrm^*,\scrq^*) \in \VSigma$ such that

\begin{enumerate}[(i)]
\item $\int_{\Sigma}\,[\eta]\,d\Delta \scrm^* - \int_{\Sigma}\beta\;d\scrq^* \geq 0,\quad \forall \ ([\eta],\beta) \in CD$.
\item $\int_{\Sigma}\,[f]\,d\Delta \scrm^* - \int_{\Sigma}g\;d\scrq^* <0,\quad \forall \ [f] \in C_0(\Sigma),\ g \in N$.
\end{enumerate}
Note that (ii) cannot be satisfied unless $\Delta \scrm^* = 0$. Therefore, since $
\beta_0$ is a member of $N$, we have
\begin{equation}\label{eq:IntBeta0Dq*>0}
\int_{\Sigma}\,\beta_0\,d\scrq^* >0.
\end{equation}
From (i) and Remark \ref{rem:CDPairsDeterminePhat} it follows that there exists $\nu \in \MSigmaPl$ such that $(\Delta \scrm^*,\scrq^* + \nu)= (0,\scrq^* + \nu)$ is a member of \scrPhat, whereupon $\scrq^* + \nu$ is a member of $\scrC^*$. Because $\beta_0$ is the reciprocal of a strong Clausius temperature scale for \theory, we must have
\begin{equation}
\int_{\Sigma}\,\beta_0\,d\,(\scrq^* + \nu) = \int_{\Sigma}\,\beta_0\,d\scrq^*  + \int_{\Sigma}\,\beta_0\,d\nu \ \leq\  0.
\end{equation}
Since $\beta_0$ takes positive values and $\nu$ is a member of \MSigmaPl, we have
\begin{equation}
 \int_{\Sigma}\,\beta_0\,d\scrq^*  \ \leq\  0,
\end{equation}
which contradicts \eqref{eq:IntBeta0Dq*>0}. Therefore $N$ contains the reciprocal of a Clausius-Duhem temperature scale. 
\end{proof}

\bibliographystyle{spmpsci}
\bibliography{MonoLibrary-Thermo}

\end{document}